\documentclass[sigconf, nonacm]{acmart}
\newcommand\vldbdoi{XX.XX/XXX.XX}
\newcommand\vldbpages{XXX-XXX}
\newcommand\vldbvolume{14}
\newcommand\vldbissue{1}
\newcommand\vldbyear{2025}
\newcommand\vldbauthors{\authors}
\newcommand\vldbtitle{\shorttitle} 
\newcommand\vldbavailabilityurl{https://github.com/zhangxiangzhi13/TRAPP}
\newcommand\vldbpagestyle{plain}

\pdfoutput=1 
\usepackage{amsmath,amsfonts}
\usepackage{algorithmic}
\usepackage{textcomp}
\usepackage{xcolor}
\usepackage{multicol}
\usepackage{amsmath}

\usepackage{enumerate}
\usepackage[linesnumbered, ruled, vlined,resetcount]{algorithm2e} 
\usepackage{graphics}
\usepackage{balance}
\usepackage{graphicx}
\usepackage{subfig}
\usepackage{amsthm,color}
\usepackage{bm}
\usepackage{listings}
\usepackage{enumitem}
\usepackage{verbatim} 
\usepackage{booktabs}
\usepackage{color}
\usepackage{bbding} 
\usepackage{stfloats} 

\usepackage{newfloat}
\usepackage{multirow}
\DeclareFloatingEnvironment[fileext=frm,placement={!ht},name=Frame]{myfloat}
\makeatletter
\newcommand{\removelatexerror}{\let\@latex@error\@gobble}
\makeatother
\newcommand{\kw}[1]{{\ensuremath{\mathsf{#1}}}\xspace}
\newcommand{\eat}[1]{}

\newcommand{\stitle}[1]{\vspace{1.6ex}\noindent{\bf #1}}
\newcommand{\etitle}[1]{\vspace{1ex}\noindent{\underline{\em #1}}}

\newcommand{\beqn}{\begin{eqnarray*}}
\newcommand{\eeqn}{\end{eqnarray*}}

\newcommand{\ie}{\emph{i.e.,}\xspace}
\newcommand{\eg}{\emph{e.g.,}\xspace}

\newcommand{\eop}{\hspace*{\fill}\mbox{$\Box$}}     
\newcounter{example}
\renewcommand{\theexample}{\arabic{example}}
\newenvironment{example}{
        \vspace{1ex}
        \refstepcounter{example}
        {\noindent\bf Example \theexample:}}{
        \eop\vspace{1ex}}

\newcounter{theorem}
\renewcommand{\thetheorem}{\arabic{theorem}}
\newenvironment{theorem}{\begin{em}
        \refstepcounter{theorem}
        {\vspace{1.5ex} \noindent\bf  Theorem  \thetheorem:}}{
        \end{em}\eop\vspace{1.5ex}} 

\newtheorem{definition}{Definition}

\newcommand{\bi}{\begin{itemize}}
\newcommand{\ei}{\end{itemize}}
        {\end{itemize}} 

\newcounter{lemma}
\renewcommand{\thelemma}{\arabic{theorem}}
\newenvironment{lemma}{\begin{em}
        \refstepcounter{theorem}
        {\vspace{1.5ex}\noindent\bf Lemma \thelemma:}}{
        \end{em}\eop\vspace{1.5ex}} 
        
\renewenvironment{proof}{
        \vspace{1ex}
        {\noindent\bf Proof:}}{\eop\vspace{1ex}}

\newcommand{\methodname}{\kw{TRAPP}}
\newcommand{\randomMethod}{\kw{Random}}
\newcommand{\allShortcutMethod}{\kw{All}}
\newcommand{\methodnameText}{\text{TRAPP}}
\newcommand{\storageOptimazation}{\text{shortcuts merger}\xspace}
\newcommand{\matchingOptimazation}{\text{shortcuts pre-sorting}}

\newcommand{\gsf}[1]{\textcolor{red}{#1}}
\newcommand{\yaof}[1]{\textcolor{cyan}{#1}}

\newcommand{\chz}[1]{\textcolor{black}{#1}}

\newcommand{\zyf}[1]{\textcolor{orange}{#1}}
\newcommand{\zyj}[1]{\textcolor{green}{#1}}
\clearpage
\begin{document}
\title{TRAPP: An Efficient Point-to-Point Path Planning Algorithm for Road Networks with Restrictions}




\author{Hanzhang Chen}
\affiliation{%
  \institution{Northeastern University}
  \city{Shenyang}
  \state{China}
}
\email{chenhz@stumail.neu.edu.cn}

\author{Xiangzhi Zhang}
\affiliation{%
  \institution{Northeastern University}
  \city{Shenyang}
  \state{China}
}
\email{zhangxz@stumail.neu.edu.cn}

\author{Shufeng Gong}
\affiliation{%
  \institution{Northeastern University}
  \city{Shenyang}
  \state{China}
}
\email{gongsf@mail.neu.edu.cn}

\author{Feng Yao}
\affiliation{%
  \institution{Northeastern University}
  \city{Shenyang}
  \state{China}
}
\email{yaofeng@stumail.neu.edu.cn}

\author{Song Yu}
\affiliation{%
  \institution{Northeastern University}
  \city{Shenyang}
  \state{China}
}
\email{yusong@stumail.neu.edu.cn}

\author{Yanfeng Zhang}
\affiliation{%
  \institution{Northeastern University}
  \city{Shenyang}
  \state{China}
}
\email{zhangyf@mail.neu.edu.cn}

\author{Ge Yu}
\affiliation{%
  \institution{Northeastern University}
  \city{Shenyang}
  \state{China}
}
\email{yuge@mail.neu.edu.cn}

\begin{abstract}


Path planning is a fundamental problem in road networks, with the goal of finding a path that optimizes objectives such as shortest distance or minimal travel time. Existing methods typically use graph indexing to ensure the efficiency of path planning. However, in real-world road networks, road segments may impose restrictions in terms of height, width, and weight. Most existing works ignore these road restrictions when building indices, which results in returning infeasible paths for vehicles. To address this, a naive approach is to build separate indices for each combination of different types of restrictions. However, this approach leads to a substantial number of indices, as the number of combinations grows explosively with the increase in different restrictions on road segments. 
In this paper, we propose a novel path planning method, \methodname (Traffic Restrictions Adaptive Path Planning algorithm), which utilizes traffic flow data from the road network to filter out rarely used road restriction combinations, retain frequently used road restriction combinations, and build indices for them. Additionally, we introduce two optimizations aimed at reducing redundant path information storage within the indices and enhancing the speed of index matching. Our experimental results on real-world road networks demonstrate that
\methodname can effectively reduce the computational and memory overhead associated with building indices while ensuring the efficiency of path planning.

\end{abstract}
\maketitle

\pagestyle{\vldbpagestyle}
\vspace{-5pt}
\section{Introduction}
\label{sec-intro}









Path planning is one of the fundamental operations in various fields, such as road networks~\cite{CRP,CH,H2H,AH,HH}, bioinformatics~\cite{bioinfo_1,bioinfo_2,bioinfo_3,bioinfo_4}, and robot navigation~\cite{robot_1,robot_2,robot_3,robot_4,robot_5}.
With the development of the transportation industry, the \eat{importance}issue of path planning on road networks has become increasingly critical.
In general, a road network can be denoted as $G=(V,E)$, where a vertex $v \in V$ represents an intersection and an edge $e \in E$ represents a road segment. Path planning in a given road network aims to find a path with the shortest distance or time between a given source vertex $s\in V$ and destination vertex $d\in V$. 
Although traditional path planning algorithms, such as Dijkstra's algorithm~\cite{Dijkstra59}, can correctly answer path queries in road networks, they need to traverse the entire road network, resulting in inefficiency for large-scale road networks. 
The A* search algorithm~\cite{AstarSearch} can reduce unnecessary traversals 
by utilizing the Euclidean distance from the current vertex to the destination vertex as a heuristic function. However, the path planning time remains significantly long. This is because when the algorithm steps into an area with dense road segments, such as the center of a city, it is required to evaluate all road segments to determine the shortest path through the dense area.

To enhance the efficiency of path planning, various index-based path planning methods are proposed in the literature~\cite{DBLP:journals/pvldb/RiceT10,chen2023fully,tan2003shortest,LSD-index,symmetry,CH,CRP,H2H,HH,AH,HiTi,TEDI}. \eat{ to address the above challenge} These methods pre-compute the shortest paths between selected vertices in a given road network, then store and maintain these paths within an index structure to reduce the time required for path planning. CRP (Customizable Route Planning)~\cite{CRP} is one of the state-of-the-art index-based path planning algorithms.
It partitions the entire road network into multiple dense cells and stores pre-computed shortest paths between the entry and exit vertices of each cell as \textit{shortcuts} (a type of index structure). 
These shortcuts are utilized to bypass the path search in dense areas.

However, in real-world road networks, there may be restrictions on factors such as height, width, and weight~\cite{RoadrestrictionsOverview_1,Roadrestrictionsweight,al2017heavy,zhu2014vehicle}.
This renders index-based methods inapplicable for directly answering path queries in road networks with restrictions. This is because a query vehicle's height, width, and weight may be larger than the corresponding restrictions represented by the shortest path, leading to infeasible shortcut for the vehicle.
%

For example, as shown in Figure \ref{fig:road_restriction}a, a typical shortcut from $v_7$ to $v_6$ that does not consider restrictions will point to the shortest path $v_7 \rightarrow v_4 \rightarrow v_6$. However, for a vehicle with a height of 2.5, a width of 2.0, and a weight of 3.0, this shortcut is infeasible because the edge $(v_4, v_6)$ has a height restriction of 1.8, as shown in Figure \ref{fig:road_restriction}b, which is lower than the height of the vehicle.

\begin{figure*}[t] 
  \centering
  \includegraphics[width=\textwidth]{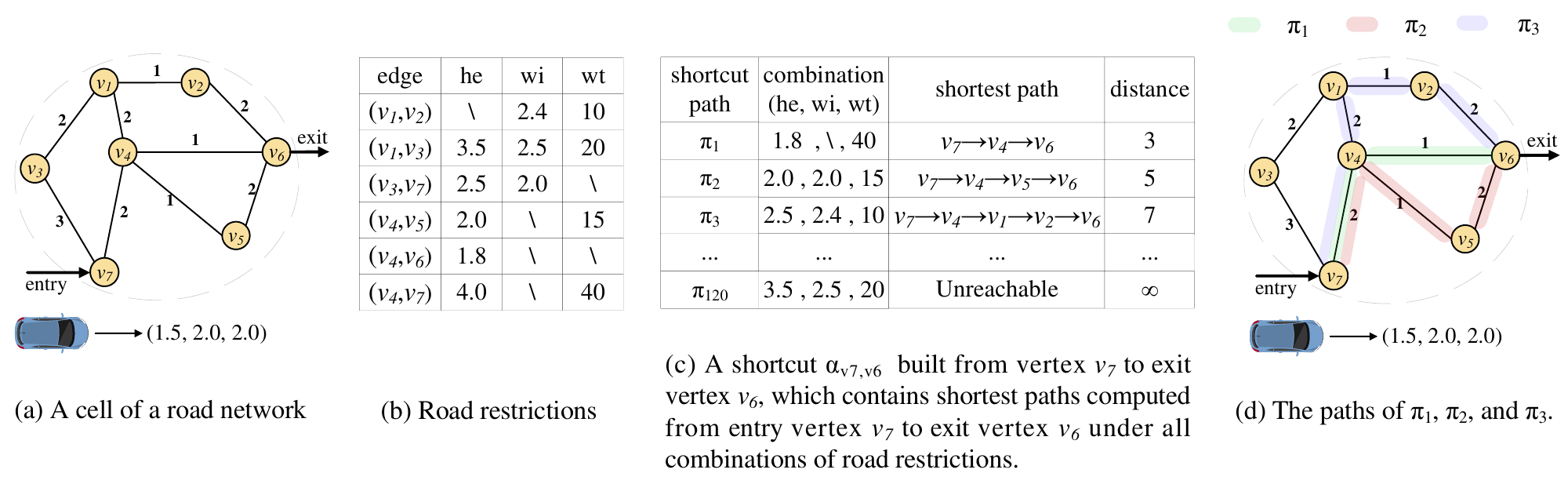} 
  \caption{
  An illustrative example to show the combinatorial explosion problem in shortcut construction for a road network cell with height (he), width (wi), and weight (wt) restrictions. For this road network cell, $v_7$ is the entry vertex and $v_6$ is the exit vertex. A vehicle with a height of 1.5, width of 2.0 and weight of 2.0 wants to travel from the entry vertex to the exit vertex.
  }
  \label{fig:road_restriction}
  \vspace{-5pt}
\end{figure*}

\eat{
\hl{For existing methods to be valid, it is necessary to build different shortcuts for different restrictions.} A naive \chz{straightforward} approach is to combine different height, width, and weight restrictions in the road network into distinct combinations and build a shortcut for each combination.
}


To enable index-based methods to plan feasible paths for vehicles, shortcuts should be built under different road restrictions. 
A naive approach is to combine different height, width, and weight restrictions in the road network into distinct combinations, compute the shortest path for each combination, and store these paths into shortcuts.
\chz{However, the vast majority of roads in the network feature various types of restrictions with differing restriction values. For example, approximately 73.8\% of road segments in China have diverse types of restrictions with varying restriction values. Additionally, China’s 1.0793 million bridges feature a wide range of weight restrictions, highlighting the diversity of restrictions across the network~\cite{ChinaRoadClassification, ChinaCountryRoadRestriction, AnhuiWeightRestriction, ChinaWeightRestrictionSign}.}
Consequently, considering all restriction combinations when building shortcuts would lead to a combinatorial explosion problem, which is intolerable in terms of both storage and computation. A 
naive way to 
reduce the number of shortcuts is to randomly build shortcuts for only a few 
of the restriction combinations. However, storing fewer shortest paths may result in vehicles obtaining sub-optimal 
paths. We illustrate the above with a detailed example.
\eat{
However, considering road restriction combinations when building shortcuts would lead to combinatorial explosion problem, which causes an excessive number of shortcuts.
A 
naive way to 
reduce the number of shortcuts is to randomly build shortcuts for only a few 
of the restriction combinations. However, building fewer shortcuts may result in vehicles obtaining sub-optimal 
paths. We illustrate this with a detailed example.
}

\begin{example}
Figure \ref{fig:road_restriction} shows how shortcuts are built in a road network with different combinations of road restrictions. 
Figure \ref{fig:road_restriction}a provides an example of a high-density cell of the road network, where $v_7$ is the entry vertex and $v_6$ is the exit vertex. Figure \ref{fig:road_restriction}b records the restrictions on different road segments, due to which, we combine different type of road restrictions and computes shortest paths for each combination to ensure that each vehicle can get a feasible shortest path.
Theoretically, there are 120 (\ie 6*4*5) shortcut paths in shortcut from $v_7$ to $v_6$, as shown in Figure \ref{fig:road_restriction}c,
which is 120 times greater than existing methods without considering road restrictions. 
Each shortcut 
corresponds to different shortest paths with different road restriction combinations. 

If we randomly compute shortest paths for only a small subset of all combinations, it is highly likely that the shortcuts containing the optimal or acceptable paths for vehicles will be discarded.
Figure \ref{fig:road_restriction}d shows the shortest paths represented by shortcut paths $\pi_1$, $\pi_2$, and $\pi_3$ under restriction combinations $(1.8, \setminus, 40)$, $(2.0, 2.0, 15)$, and $(2.5, 2.4, 10)$.
Assuming a vehicle with a height of 1.5, width of 2.0, and weight of 2.0 intends to travel from $v_7$ to $v_6$, $S_1$ represents its global shortest path. 
If we only store 
shortcuts $\pi_2$ and $\pi_3$, as shown in Figure \ref{fig:road_restriction}d, 
the vehicle has to travel from $v_7$ to $v_6$ via the path represented by $\pi_2$, \ie $v_7\to v_4\to v_2\to v_6$. Because $\pi_2$ is shorter than $\pi_3$. 
\end{example}


\eat{
Moreover, due to the road restrictions, more shortcuts must be built to ensure the availability of the planned roads \gsf{why more shortcuts?}.
}




\stitle{Motivation}.\eat{\gsf{Too wordy, summarize the above and identify the problems you should solve.}}
Based on the above analysis, we find that the combinatorial explosion problem 
necessitates extensive pre-computation of shortcuts and results in significant memory consumption in existing methods.
However, having fewer shortcuts between entry and exit vertices increases the likelihood of missing the optimal or an acceptable path, leading to potentially worse solutions.
\eat{
However, the fewer shortcuts between entry and exit vertices, the worse solution 
we may 
find since the optimal or acceptable path may be lost with a higher probability. 
}
This motivates us to design a method to build fewer shortcuts in the road network and still allow most vehicles to find optimal 
or acceptable paths. 

\stitle{Our solution}. 
To address the combinatorial explosion problem in shortcut construction caused by the restrictions of road segments, 
we propose a novel path planning method for road networks with restrictions, \methodname. 
\methodname 
utilizes vehicles' traffic flow data in the road network to filter out rarely used restriction combinations while preserving frequently used ones, which reduces 
the number of shortcuts 
to a reasonable range meanwhile ensuring the quality of path planning.
Additionally, we introduce two optimizations: \storageOptimazation \ to reduce memory usage and \matchingOptimazation \ to expedite shortcuts matching during path planning.
\eat{
To address the issue of excessive indexing mentioned above, we propose a novel method, \methodname.
The core idea of \methodname is to use vehicles' information of traffic flow 
of the road network to filter out rarely used restriction combinations and save the frequently used restriction combinations. And then, we build shortcuts only for the remaining combinations, this way we no longer need to build shortcuts for every possible restriction combinations of the road network. So that we control the number of shortcuts within a reasonable range and ensure the quality of path planning.
}

In summary, our contributions are as follows:
\begin{itemize}[leftmargin=*, topsep=0pt]
    \item \textbf{Studying and analyzing path planning in road networks with restrictions.}
    We formally study the problem of path planning in road networks with restrictions, providing a detailed definition of the problem. We critically analyze the limitations of existing methods when applied to restricted road networks and propose a novel approach to address these challenges.

    \item \textbf{Efficient path planning method.}
     We propose an efficient path planning method tailored for road networks with restrictions. This method significantly reduces the computational and memory overhead associated with building shortcut paths while enabling most vehicles to find optimal or acceptable paths in a short time.     

    \item \textbf{Optimization of shortcut matching speed and shortcut storage.}
     We propose an optimized method for shortcut storage that significantly reduces memory usage by minimizing redundant path storage. Additionally, we have designed a method for quickly shortcut matching, which substantially decreases the matching time and thus accelerates path query processing.

    \item \textbf{Extensive experiments are conducted on real-world road networks.} 
    We conducted extensive experiments on six real-world road networks. The experimental results demonstrate that our proposed path planning method \methodname achieves excellent path planning performance while ensuring high efficiency.

\end{itemize}

\stitle{Organization}. The rest of this paper is organized as follows. We provide preliminaries in section \ref{sec:prelim}. In section \ref{sec:statement}, we analyze the path planning problem in the road networks with restrictions. We provide a detailed introduction to our proposed method in section \ref{sec:method}. Experimental results are reported in section \ref{sec:experiment} and the existing path planning methods are summarized in section \ref{sec:related_work}. Finally, we conclude the paper in section \ref{sec:conclusion}.
\section{Preliminaries}
\label{sec:prelim}

In this section, we first formally introduce some basic concepts related to our work. 

\begin{table}[ht]
\centering
\label{tab:notations}
\caption{Frequently used notations}
\vspace{-0.15in}
\renewcommand{\arraystretch}{1.4}
\begin{tabular}{p{3.0cm}|p{5.0cm}}
\toprule
\textbf{Notation} & \textbf{Description} \\
\midrule
$G=(V, E, R_E, L_E)$ & Road network with restrictions \\
\hline
$he,wi,wt$ & Height, width and weight \\
\hline
$R_E=\{R^{he}_E, R^{wi}_E, R^{wt}_E\}$ & Function to get restrictions \\
\hline
$rc=(he,wi,wt)$ & Road restriction combination \\
\hline
$c=(he,wi,wt)$ & Vehicle \\
\hline
$rv=<he,wi,wt>$ & Representation vector \\
\hline
$\alpha_{s,d},S$ & Shortcut from $s$ to $d$ and a shortcut set \\
\hline

$\pi_{s,d}^{rc}$ & Shortcut path from $s$ to $d$ under $rc$  \\
\hline
$dist()$ & Function to get distance\\
\hline
$C=(V,E,R_E, L_E)$ & Cell of a road network \\
\bottomrule
\end{tabular}
\end{table}

\stitle{Road Network with Restrictions}.
Let $G=(V, E, R_E, L_E)$ be a road network with restrictions where $V$ is a finite set of vertices that represents the intersections, $E = \{(u, v)\} \subseteq V \times V$ is a set of edges which represents road segments. $R_E=\{R^{he}_E, R^{wi}_E, R^{wt}_E\}$ 
contains three restriction functions such that each edge $(u, v) \in E$ carries the \textit{height restriction} $R^{he}_{u,v}$, \textit{width restriction} $R_{u,v}^{wi}$, and \textit{weight restriction} $R_{u,v}^{wt}$. $L_E$ is a length function that each edge $(u,v)\in E$ carries a length value $l_{u,v}$. 

\stitle{Vehicle}. A vehicle is defined as a triplet $c = (he,wi,wt)$ with three attributes, where $he$, $wi$, and $wt$ represent the vehicle's height, width, and weight, respectively.
\eat{
\yaof{The vehicle's attributes corresponding to road restrictions are defined as a triplet $c = (he,wi,wt)$, where $he$, $wi$, and $wt$ represent the vehicle's height, width, and weight, respectively.}
}

\stitle{Feasible Vehicle Path and Its Distance}.
Given a path from $v_0$ to $v_k$ that is denoted as a sequence of vertices $p_{v_0,v_k}=(v_0,v_1,\dots,v_k)$ 
and a vehicle $c=(he,wi,wt)$, the path $p_{v_0,v_k}$ 
is a \textit{feasible vehicle path} for the vehicle $c$ if and only if 
$$
    c.he \leq R^{he}_{v_i,v_{i+1}},
$$
$$
    c.wi \leq R^{wi}_{v_i,v_{i+1}},
$$
$$
    c.wt \leq R^{wt}_{v_i,v_{i+1}},
$$
for each edge $(v_i,v_{i+1}) \in p_{v_0,v_k}$. The distance of $p_{v_0,v_k}$ is the sum of the length of all road segments on $p_{v_0,v_k}$, \ie $dist(p) = \sum_{i=0}^{k-1} l_{i,i+1}$.


\stitle{Shortest Feasible Vehicle Path.}
Given a set of feasible vehicle paths from vertex $v_0$ to $v_k$, $P_{v_0,v_k}=\{p^1_{v_0, v_k}, \cdots, p^N_{v_0, v_k}\}$ for the vehicle $c=(he,wi,wt)$, 
the \textit{shortest feasible vehicle path} $p^*_{v_0,v_k}$ from $v_0$ to $v_k$ 
for $c$ is defined as 
$$
dist(p^*_{v_0,v_k}) \leq dist(p^i_{v_0, v_k}), 1\leq i \leq N.
$$


\etitle{Path Planning.} 
Given a road network with restrictions $G = (V, E, R_E, L_E)$ and a vehicle $c=(he,wi,wt)$, path planning aims to find a \textit{shortest feasible vehicle path} $p^*_{s,d}$ from a given source vertice $s$ to a destination vertice $d$ for the vehicle $c$. 

As we discussed previously, many methods~\cite{CRP,CH,HH,AH} accelerate path planning by building shortcuts in dense areas of road networks, \ie precomputing the shortest path from an entry to an exit of the dense area, thereby reducing the online search space. 
However, traditional shortcuts do not consider road restrictions when they are built, and each shortcut represents 
only one single shortest path. In a road network with restrictions, the shortcut 
should store multiple shortest paths for different road restrictions combinations, so that we can find the shortest \textit{feasible vehicle path} for different vehicles with different heights, widths, and weights. Based on this, we formally define the 
shortcuts on road networks with restrictions. Before that, we first provide some basic concepts, \ie entry/exit vertices, feasible shortcut paths, and the shortest feasible shortcut path.

\stitle{Entry/Exit vertices}. Let $G_i=(V_i, E_i, R_{E_i}, L_{E_i})$ is a subgraph of a road network with restriction $G$, where $E_i\subseteq V_i \times V_i \cap E$. The vertex $v\in V_i$ is an entry (exit) vertex if it has an incoming edge $(u,v)\in E\setminus E_i$ \big(outgoing edge $(v,w)\in E\setminus E_i$\big)  and $u\not\in E_i$ ($w\not\in E_i$).  

\stitle{Feasible Shortcut Path}.
Given a shortcut path $\pi_{v_0,v_k}=(v_0, \cdots, v_k)$ from the entry vertex $v_0$ to exit vertex $v_k$ of a subgraph and a restriction combination $rc=(he,wi,wt)$, the path $\pi_{v_0,v_k}$ is a \textit{feasible shortcut path} under restriction $rc$ if and only if
$$
    rc.he \leq R^{he}_{v_i,v_{i+1}},
$$
$$
    rc.wi \leq R^{wi}_{v_i,v_{i+1}},
$$
$$
    rc.wt \leq R^{wt}_{v_i,v_{i+1}},
$$
for each edge $(v_i,v_{i+1}) \in \pi_{v_0,v_k}$.

\stitle{Shortest Feasible Shortcut Path}. Given a set of feasible shortcut paths from $s$ to $d$, $\Pi_{s,d} = \{\pi^1_{s, d}, \cdots, \pi^N_{s, d}\}$ for $rc=(he,wi,wt)$, the \textit{shortest feasible shortcut path} is defined as
$$
dist(\pi^{rc}_{v_0,v_k}) \leq dist(\pi^i_{v_0, v_k}), 1\leq i \leq N.
$$

\stitle{Domination}. 
Given a vehicle $c=(he,wi,wt)$ and a restriction combination $rc=(he,wi,wt)$, 
$c$ is dominated by $rc$ if and only if
$$
c.he \leq rc.he,
$$
$$
c.wi \leq rc.wi,
$$
$$
c.wt \leq rc.wt.
$$
\begin{lemma}
\label{lem:dominate}
    Given a vehicle $c=(he,wi,wt)$, a restriction combination $rc=(he,wi,wt)$, and a feasible shortcut path $\pi$ for $rc$,
    if 
    $c$ is dominated by 
    $rc$, then $\pi$ a feasible path for $c$.
\end{lemma}

\begin{proof}
    Given a feasible shortcut path $\pi_{v_0,v_k}=(v_0, \cdots, v_k)$ for $rc=(he,wi,wt)$ and a vehicle $c=(he,wi,wt)$ dominated by $rc$. We have $rc.he \leq R^{he}_{v_i,v_{i+1}}$, $rc.wi \leq R^{wi}_{v_i,v_{i+1}}$ and $rc.wt \leq R^{wt}_{v_i,v_{i+1}}$ for each edge $(v_i,v_{i+1}) \in \pi_{v_0,v_k}$. Since $c.he \leq rc.he$, $c.wi \leq rc.wi$ and $c.wt \leq rc.wt$. Then we have $c.he \leq R^{he}_{v_i,v_{i+1}}$, $c.wi \leq R^{wi}_{v_i,v_{i+1}}$ and $c.wt \leq R^{wt}_{v_i,v_{i+1}}$ for each edge $(v_i,v_{i+1}) \in \pi_{v_0,v_k}$, \ie $\pi_{v_0,v_k}$ is a feasible path for $c$.   
\end{proof}

\stitle{Shortcut on Restriction Road Network}. 
Given a dense area of the road network with restrictions, where the dense area is a subgraph of the road network and is called a dense cell, the \textit{shortcut} $\alpha_{s,d} = \{\pi_{s,d}^{rc_0}, \pi_{s,d}^{rc_1}, \dots, \pi_{s,d}^{rc_k}\}$ stores a set of shortest paths from $s$ which is the entry vertex of $C$ to $d$ which is the entry vertex of $C$, each computed under different road restriction combinations $rc$. Here, $rc$ is represented as $rc=(he,wi,wt)$, where $he$, $wi$, and $wt$ denote height restriction, width restriction, and weight restriction, respectively.

\begin{example}
As illustrated in Figure \ref{fig:road_restriction}c, a shortcut $\alpha_{v_7,v_6} = \{\pi_0, \pi_1, \dots,\\ \pi_{120}\}$ from vertex $v_7$ to vertex $v_6$ is built.

As illustrated in Figure \ref{fig:road_restriction}c, \{$\pi_1$, $\pi_2$, \dots $\pi_{120}$\} are the shortcut paths computed from $v_7$ to $v_6$, where $v_7$ is the source vertex, and $v_6$ is the destination vertex. Taking $\pi_3$ as an example, $\pi_3$ is the shortest path computed under the restriction combination $(2.5, 2.4, 10)$, meaning that the height restriction of $\pi_3$ (\ie $R^{hi}(\pi_3)$) is 2.5, the width restriction (\ie $R^{wi}(\pi_3)$) is 2.4, and the weight restriction (\ie $R^{wt}(\pi_3)$) is 10.0. The shortcut path $\pi_3$ consists of the road segments $(v_7, v_4)$, $(v_4, v_1)$, $(v_1, v_2)$, and $(v_2, v_6)$, with a total distance of 4.
\end{example}

\section{Problem Statement}
\label{sec:statement}
In this section, we first introduce the state-of-the-art two-stage framework for path planning. Then we analyze the problem of the combinatorial explosion of this framework on the road networks with restrictions, which leads to a surge in memory and computation. Based on that, we introduce the goal of this paper.


\subsection{Two-stage path planning framework} 
The existing state-of-the-art path planning algorithms~\cite{CRP,HH,CH} employ a two-stage framework to search the shortest path for the given source and destination vertices. 
This framework contains an offline preprocessing stage and an online query stage. 

In the offline preprocessing phase, it first identifies the dense areas of the road network, called dense cells of the road network, which is done by the community detection algorithms, such as PUNCH~\cite{PUNCH}, METIS~\cite{METIS}, KaPPa\cite{KaPPa} and SCOTCH~\cite{SCOTCH}. Then it builds shortcuts from the entry vertex to the exit vertex of the dense cell. In the online querying phase, it employs an existing smart path planning algorithm, such as Dijkstra~\cite{Dijkstra59} and A* search~\cite{AstarSearch}, to find the shortest path. When entering a dense cell, the path planning algorithm directly visits the shortest path represented by the shortcuts to avoid extensive searches within the cell, thereby significantly enhancing the efficiency of path planning.

As shown in Figure \ref{fig:overview}, in the preprocessing stage, it first partitions the road network $G$ shown in Figure \ref{fig:overview}a in into multiple cells, where different background colors represent different cells.  
Next, it computes the shortest paths between the entry and exit vertices within each cell and stores these paths in shortcuts, as illustrated in Figure \ref{fig:overview}c. In this figure, red dashed lines represent the shortcuts built between the entry and exit vertices, and red solid lines indicate the actual shortest paths recorded by these shortcuts. As shown in Figure \ref{fig:overview}d, with the help of shortcuts, the search space is significantly reduced.

\begin{figure}
    \centering
    \includegraphics[width=3.0in]{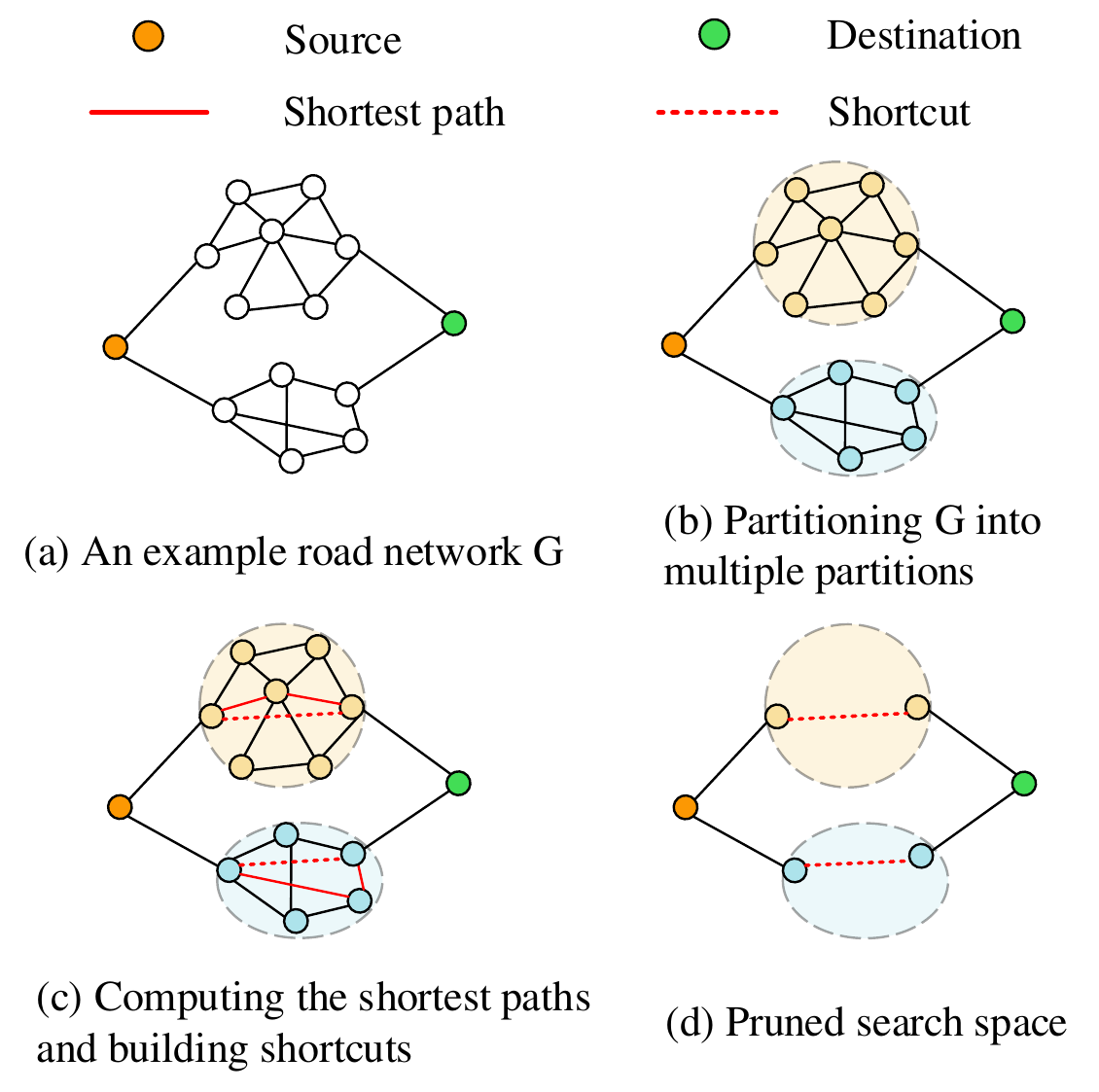}
    \caption{An example to illustrate the process of shortcut building by existing methods. 
    The vertices of the road network represent the intersections, the edges represent roads, and the shortcuts represent the record of pre-computed shortest paths. 
    \eat{
    \zyf{use different color to differentiate shortest path from shortcut. not show the search process, should compare the two candidate paths. why are source and dest not in cells? (d) pruned search space not understand.}
    }
    }
    \label{fig:overview}
    \vspace{-0.30in}
\end{figure}




\stitle{Drawbacks}. The traditional two-phase framework introduced above does not consider the restrictions within the road network, and each shortcut only stores a single shortest path that is treated as feasible for all vehicles. However, in road networks with restrictions, some edges may have restrictions, as introduced in Section \ref{sec:prelim}. This necessitates the pre-computation of multiple \textit{feasible shortcut paths} under different restriction combinations, allowing vehicles with different height, width and weight to find their \textit{feasible shortest paths}. Unfortunately, dense cells contain many edges and, consequently, a variety of road restrictions with differing values. This results in an exponential increase in the number of restriction combinations, \ie combinatorial explosion problem. To illustrate, consider a cell with $N_{he}$ different height restrictions, $N_{wi}$ different width restrictions, and $N_{wt}$ different weight restrictions. The number of all restriction combinations in this cell is $N_{he}*N_{wi}*N_{wt}$. For example, as shown in Figure \ref{fig:road_restriction}, the cell depicted in Figure Figure \ref{fig:road_restriction}a has 120 restriction combinations (\ie 6*4*5). 

Therefore, computing feasible shortcut paths for all road restrictions incurs significant computational and storage overhead. Additionally, this approach greatly increases the time required for online queries. This is because, for a vehicle 
$c$, if the restriction combination $rc$ of the shortest feasible shortcut path in the shortcut dominates $c$, then that path is a feasible path for $c$ (as demonstrated in Lemma \ref{lem:dominate}). The more restriction combinations there are, the more feasible paths are stored in the shortcut for $c$, leading to an increased search effort to find the shortest feasible vehicle path for $c$. In other words, the combinatorial explosion problem not only results in significant computational and storage overhead but also increases the online planning time.


\stitle{Naive Solution}. To avoid the calculation and storage of a large number of shortest feasible shortcut paths, a naive method is to randomly select a part of the restriction combinations and calculate their shortest feasible shortcut paths. Since some shortest feasible shortcut paths for some restriction combinations are discarded, the planned path may not be the shortest. This is because if the shortest feasible shortcut path, for a given vehicle, is discarded on a shortcut, this vehicle will search for another shortest feasible shortcut path for wider restriction so that the vehicle can find a feasible path. However, the shortest feasible shortcut path of wider restriction always is longer. Therefore, the path planned by the shortest feasible shortcut paths that are randomly saved in the shortcut is not the shortest feasible vehicle path, but longer than the actual shortest path. The fewer the shortest feasible shortcut paths saved in the shortcut, the longer the planned path will be compared to the shortest feasible vehicle path. 

Only calculating and saving a small number of shortest feasible shortcut paths will lead to a longer shortest feasible vehicle path. For a given number of shortest feasible shortcut paths, is there any shortest feasible shortcut paths selection method that minimizes the shortest feasible vehicle path for all vehicles? 

\stitle{Our Goal}. We aim to design a shortest feasible shortcut path selection method $\mathcal{S}$, such that for a vehicle set $C=\{c_i\}$, the sum of the length of the shortest feasible vehicle path is minimized, \ie
\begin{equation}
    minimize \sum^{N}_{i=1}dist(\mathcal{P}_{c_i}(\mathcal{S}(\Pi)))
\end{equation}
where $\mathcal{P}_{c_i}$ is the shortest path planning method, \eg $A*$, that return the shortest feasible path.

\section{\methodnameText \text{ } algorithm}
\label{sec:method}

This section provides a detailed introduction to our proposed path planning method, \methodname. The core idea of \methodname is to filter road restriction combinations using traffic flow data from the road network. We retain the combinations used by the majority of vehicles and discard those used by only a few. Building shortcuts solely for the retained combinations significantly reduces the number of shortcuts, thereby greatly decreasing memory usage and computational overhead.

\stitle{\chz{Representative Vehicle Dimension Mining}}.
\chz{The height, width, and weight of vehicles in real-world road networks exhibit distinct clustering trends within localized regions~\cite{ChinaVehicleClassification}. Based on this observation, we aim to identify the primary distribution ranges of these attributes to derive representative vehicle dimensions. These representative dimensions facilitate efficient filtering of road restriction combinations in subsequent tasks.}

\chz{To achieve this, we first partition the road network 
$G = (V, E, R_E, L_E)$ into dense cells 
$C_i = (V_i, E_i, R_E, L_E)$ using the PUNCH algorithm~\cite{PUNCH}. 
This partitioning confines the analysis of traffic flow data to individual cells, 
allowing the extraction of vehicle characteristics—such as height, width, and weight—that accurately reflect the traffic patterns within each region. 
Within each cell, each vehicle is abstracted as a three-dimensional vector 
$v = \langle he, wi, wt \rangle$, 
where $he$, $wi$, and $wt$ represent the vehicle’s height, width, and weight, respectively.
}

\chz{
To extract representative vehicle dimensions, we apply the $K$-means clustering algorithm to group vehicles with similar dimensions. The traffic flow data is partitioned into $K$ clusters \( \{\mathcal{VG}_1, \mathcal{VG}_2,\\ \dots, \mathcal{VG}_K\} \), where $K$ is empirically set to 30 based on observed trade-offs between computational efficiency and capturing representative dimension ranges. For each cluster $\mathcal{VG}_i$, we compute a representation vector $rv_i = \langle he_{\max}, wi_{\max}, wt_{\max} \rangle$, where $he_{\max}$, $wi_{\max}$, and $wt_{\max}$ represent the maximum height, width, and weight of vehicles within the cluster, respectively. 
}
\chz{
The underlying idea is to leverage the height, width, and weight attributes of the vector to construct a \textit{restriction combination} and build a corresponding \textit{shortcut path}. This ensures that every vehicle within the cluster can efficiently identify a \textit{feasible shortcut path} through the built \textit{shortcut path}, thereby reducing both computational and storage overhead due to the limited number of \textit{shortcut paths} maintained. Subsequently, we will formally define the representation vector and provide a rigorous theoretical foundation and proof for the proposed idea.
}

\chz{
\begin{definition}[representation vector]
\label{def:representation_vector}
Given a vehicle cluster $\mathcal{VG}$, the \textit{representation vector} is defined as a three-dimensional vector $rv = \langle he_{\max}, wi_{\max}, wt_{\max} \rangle$, where $he_{\max}$, $wi_{\max}$, and $wt_{\max}$ represent the maximum height, width, and weight of vehicles within the cluster $\mathcal{VG}$, respectively. 
\end{definition}
}
\chz{
\begin{lemma}
\label{lem:representation_vector_feasible}
For a given vehicle cluster $\mathcal{VG}$, the shortcut path built under the height $he$, width $wi$, and weight $wt$ of the cluster's representation vector $rv = <he, wi, wt>$ is guaranteed to be a \textit{feasible vehicle path} for all vehicles $c \in \mathcal{VG}$.
\end{lemma}
}

\chz{
\begin{proof}
    Consider a vehicle cluster \(\mathcal{VG}\) with its representation vector \(rv = \langle he, wi, wt \rangle\). Let the shortcut path \(\pi_{v_0, v_k} = (v_0, \ldots, v_k)\) be built based on the \textit{restriction combination} \(rc = ( rv.he, rv.wi, rv.wt )\). 
    For each vehicle \(c = ( he, wi, wt ) \in \mathcal{VG}\), by definition of the representation vector, we have \(c.he \leq rv.he\), \(c.wi \leq rv.wi\), and \(c.wt \leq rv.wt\). This implies that the vehicle \(c\) is \textbf{dominated by} \(rc\).
     By Lemma~\ref{lem:dominate}, \(\pi_{v_0, v_k}\) is a \textit{feasible vehicle path} for every vehicle \(c \in \mathcal{VG}\).
\end{proof}
}

\chz{
The process for computing the \textit{representation vectors}, as detailed in Algorithm~\ref{alg:Getting_representation_vectors}, takes the cells of the road network $\{C_1, C_2, \ldots, C_n\}$ and the vehicle groups $\{\mathcal{VG}_1, \mathcal{VG}_2, \ldots, \mathcal{VG}_K\}$ as input. In lines 3--9, the algorithm iteratively processes each cell \( C_i \) in the road network. For each vehicle group $\mathcal{VG}_k$ within a cell, lines 6--7 compute the maximum values of height, width, and weight across all vehicles in the group to form a representative vector \( rv_k \), which summarizes the key attributes of the group. Finally, in line 10, the algorithm aggregates all representation vectors into the global set \( RV \), providing a compact summary of vehicle characteristics to facilitate further filtering and optimization.
}

\begin{algorithm}
  \caption{Computing representation vectors} \label{alg:Getting_representation_vectors}
    \KwIn {$\mathcal{VG}$:set of vehicle groups$\{\mathcal{VG}_1,\mathcal{VG}_2,\ldots,\mathcal{VG}_n\}$\\
    $C$: set of road network cells$\{C_1,C_2,\ldots,C_n\}$} 
    \KwOut{RV: set of representation vectors \\$\{RV_1,RV_2,\dots,RV_n\}$}

    $RV \gets \varnothing$\\
    \For{each $C_i$ $\in$ $C$}{

        $RV_i \gets \varnothing$\\
        \For{each $\mathcal{VG}_i$ $\in$ $\mathcal{VG}$}{
            $rv_i \gets  \varnothing$\\
            \For{$\mathcal{T} \in$ $\{he,wi,wt\}$}{
                $rv_i.\mathcal{T} \gets \max\{c_j.\mathcal{T} \mid c_j \in      C_i\}$       
            }
             $RV_i \gets$ $RV_i$ $\cup$ $\{rv_i\}$
        }
        $RV \gets$ $RV \cup RV_i$
    }
    
    \KwRet{RV}
    
\end{algorithm}

\stitle{Cell-aware Refinement}
\eat{
At this stage, we integrate the previously obtained \textit{representation vectors} with the actual road restriction combinations within each cell of the road network. Consequently, we map the concentrated range information of vehicle height, width, and weight to the corresponding road restrictions. This process achieves the preliminary filtering of road restriction combinations, retaining those frequently utilized by the majority of vehicles. Specifically, the method for integrating road restrictions with representation vectors involves refining each component of the representation vectors based on the actual road restrictions.
}
\chz{
In this stage, we integrate the previously computed \textit{representation vectors} with the actual road restriction combinations specific to each cell in the road network, tailoring this integration to the unique restrictions of each cell. By mapping the aggregated ranges of vehicle height, width, and weight to the corresponding road restrictions, our method enables efficient filtering of road restriction combinations, significantly reducing the number of stored restriction combinations and the corresponding index size. This process not only optimizes storage requirements but also simplifies subsequent queries, thereby enhancing the overall system performance.
}

\chz{
Before delving into the specific methodology, we first introduce the concepts related to the mapping of \textit{representation vectors} and the objectives of filtering road restriction combinations.
}

\chz{
\begin{definition}[representation vector mapping]
\label{def:representation_vector_mapping}
Consider a cell \(C\) in a road network, with the set of \textit{restriction combinations} \(RC = \{rc_1, rc_2, \dots, rc_n\}\) and the set of \textit{representation vectors} \(RV = \{rv_1, rv_2, \dots, rv_k\}\). A representation vector mapping is defined as the process of mapping each \(rv_i \in RV\) to a unique \(rc_j \in RC\), forming a one-to-one correspondence \(f: RV \to RC\) such that for every \(rv_i \in RV\), there exists a unique \(rc_j \in RC\).
\end{definition}
}

\chz{
In this paper, we define the mapping method $f$ for the representation vector as mapping each \(rv\) to the \(rc\) in the set of \textit{road restrictions} \(RC\) with the smallest Euclidean norm with respect to the current \(rv\). This mapping method is to get the most similar restriction combination for representation vector. Formally, the mapping method is defined as:
\begin{equation}
\label{equ:vector_mapping}
    f(rv) = \arg\min_{rc \in RC} \| rv - rc \|_2
\end{equation}
where \(\| rv - rc \|_2\) represents the Euclidean norm between \(rv\) and \(rc\).
}
\chz{
Thus, in a given cell $C$, our objection for the mapping function is defined as:
\begin{equation}
\label{equ:vector_mapping_objection}
minimize \sum^{K}_{i=1}\| rv_i - f(rv_i) \|_2
\end{equation}
}

\chz{
}

\chz{
A naive approach to solving the mapping problem involves constructing all possible road restriction combinations within a given cell and iterating through them to determine the corresponding \(rc\) for each representation vector \(rv\). While conceptually straightforward, this approach suffers from significant computational inefficiency, as its time complexity scales as \(O(K\cdot N_{he} \cdot N_{wi}\cdot N_{wt})\) for a single cell, where $K$ denotes the number of representation vectors, and $N_{he}$, $N_{wi}$, and $N_{wt}$ represent the number of height, width, and weight restrictions, respectively. Such a prohibitive cost renders this approach impractical for large-scale road networks.
To overcome these limitations, we introduce a novel and computationally efficient method that eliminates the need to exhaustively generate and evaluate all road restriction combinations. By reconsidering and systematically decomposing Equation~\ref{equ:vector_mapping}, we identify that the mapping condition for each attribute \(\mathcal{T} \in \{he, wi, wt\}\) can be independently optimized. This reformulation enables the mapping process to be expressed as the minimization of the \(L_2\)-norm for each attribute:
\begin{equation}
\label{equ:type_representation_vector_mapping}
    f(rv).\mathcal{T} = \arg\min_{rc \in RC} \| rv.\mathcal{T} - rc.\mathcal{T} \|_2
\end{equation}
, where \(rv.\mathcal{T}\) and \(rc.\mathcal{T}\) denote the values of the attribute \(\mathcal{T}\) in the representation vector and road restriction combination, respectively.
This decomposition fundamentally reduces the computational overhead by allowing independent processing of each road restriction type, obviating the need for pre-computing all possible combinations.
}

\chz{
Based on the above analysis, we propose an efficient algorithm for computing the mapping process, as shown in Algorithm~\ref{alg:Efficient_Mapping_Processing}. The algorithm begins by iterating through each road network cell $C_i \in C$, aggregating road restrictions from all edges within the cell (Lines 1--6). For each edge $(u, v) \in E_i$ and each attribute $\mathcal{T} \in \{he, wi, wt\}$ (representing height, width, and weight), the algorithm collects non-empty road restrictions $R_{u,v}^\mathcal{T}$ into the corresponding restriction set $\mathcal{R}_i^\mathcal{T}$ for the cell. Once all restrictions are aggregated, the algorithm sorts each restriction set $\mathcal{R}_i^\mathcal{T}$ in ascending order to enable efficient mapping queries (Lines 7--8). 
Next, the algorithm processes the representation vectors $RV_i$ associated with $C_i$. For each vector $rv \in RV_i$, it iteratively maps each attribute $\mathcal{T}$ to the closest road restriction in $\mathcal{R}_i^\mathcal{T}$ based on Equation~\ref{equ:type_representation_vector_mapping} (Lines 9--13). 
It is evident that the time complexity of the proposed algorithm is $O(K \cdot (N_{he} + N_{wi} + N_{wt}))$, which represents a significant reduction compared to the naive approach.
}

\begin{algorithm}
  \caption{Efficient Mapping Processing for Representation Vectors} 
  \label{alg:Efficient_Mapping_Processing}
  \KwIn{$C$: set of road network cells $\{C_1, C_2, \ldots, C_n\}$\\
  RV: set of representation vectors $\{RV_1, RV_2, \dots, RV_n\}$}
  
  \For{each $C_i \in C$}{
    \For{each $(u,v) \in E_i$}{
      \For{$\mathcal{T} \in \{he, wi, wt\}$}{
        \If{$R_{u,v}^\mathcal{T} \neq \varnothing$}{
          $\mathcal{R}_i^\mathcal{T} \gets \mathcal{R}_i^\mathcal{T} \cup R_{u,v}^\mathcal{T}$ \\ //store restrictions in corresponding sets
        }
      }
    }

    \For{$\mathcal{T} \in \{he, wi, wt\}$}{
      sort $\mathcal{R}_i^\mathcal{T}$ in ascending order
    }

    \For{each $RV_i \in RV$}{
      \For{each $rv \in RV_i$}{
        \For{$\mathcal{T} \in \{he, wi, wt\}$}{
          $rv.\mathcal{T} \gets \arg\min\limits_{r \in \mathcal{R}_i^{\mathcal{T}}} \| rv.\mathcal{T} - r \|_2$
          \\ //map $rv$ attributes based on Equation~\ref{equ:type_representation_vector_mapping}
        }
      }
    }
  }
\end{algorithm}

\stitle{Combination Rematch.}
\chz{
In real-world scenarios, most vehicles exhibit a consistent proportional relationship among height, width, and weight, with these ratios typically falling within a certain range. However, we observe that some vehicles deviate significantly from this standard pattern. For example, certain trucks may have similar height and width but vary greatly in weight, while some Special Purpose Vehicles show disproportionate ratios of height, width, and weight compared to typical passenger vehicles. These vehicles tend to result in longer paths in path planning. To substantiate this, we will first introduce key concepts and provide the corresponding theoretical framework and proofs.
}

\chz{
\begin{definition}[Feasible Edge Set]
\label{def:feasible_edge_set}
In a given cell of a road network $C=(V,E,R_E, L_E)$, the set of feasible edges \( E_{rc} \) for a restriction combination \( rc = (he, wi, wt) \) is defined as:
\[
E_{rc} = \{ e \in E \mid rc.he \leq R^{he}_e, \, rc.wi \leq R^{wi}_e, \, rc.wt \leq R^{wt}_e \}.
\]
\end{definition}
}

\chz{
\begin{lemma}
\label{lem:rematch_reason}
    In a cell \( C \) of a road network, consider two shortcut paths 
    \( \pi_{v_0,v_k}^{rc_i} = (v_0, \dots, v_k) \) and 
    \( \pi_{v_0,v_k}^{rc_j} = (v_0, \dots, v_k) \) with \( i \neq j \), 
    that connect the entry vertex \( v_0 \) to the exit vertex \( v_k \). 
    If \( rc_i.he \leq rc_j.he \), \( rc_i.wi \leq rc_j.wi \), and 
    \( rc_i.wt \leq rc_j.wt \), then:
    \[
    dist(\pi_{v_0,v_k}^{rc_i}) \leq dist(\pi_{v_0,v_k}^{rc_j}).
    \]
\end{lemma}
}

\chz{
\begin{proof}
Given that \( rc_i.he \leq rc_j.he \), \( rc_i.wi \leq rc_j.wi \), and \( rc_i.wt \leq rc_j.wt \), by Definition~\ref{def:feasible_edge_set}, we have \( E_{rc_j} \subseteq E_{rc_i} \). Since \( \pi_{v_0,v_k}^{rc_i} \) and \( \pi_{v_0,v_k}^{rc_j} \) are the shortest paths from \( v_0 \) to \( v_k \) under \( rc_i \) and \( rc_j \), respectively, and \( E_{rc_j} \subseteq E_{rc_i} \), the path \( \pi_{v_0,v_k}^{rc_i} \) can only use the same or more feasible edges than \( \pi_{v_0,v_k}^{rc_j} \). Therefore, \( dist(\pi_{v_0,v_k}^{rc_i}) \leq dist(\pi_{v_0,v_k}^{rc_j}) \).
\end{proof}
}

Building on Lemma~\ref{lem:rematch_reason}, it is clear that these vehicles encounter significant challenges in obtaining \textit{shortest feasible vehicle paths} during shortcut matching. This often leads to suboptimal or unattainable paths. Consequently, targeted optimization strategies for these vehicles can effectively enhance the path planning system.

These vehicles can be broadly classified into two categories: The first category includes vehicles where one dimension (e.g., height, width, or weight) is significantly larger than the others. The second category consists of vehicles where one dimension is significantly smaller than the others, such as vehicles with weight much lower than their height and width. This imbalance reduces the effectiveness of path planning.

\begin{figure}[htbp]
    \centering
    \includegraphics[width=3.35in]{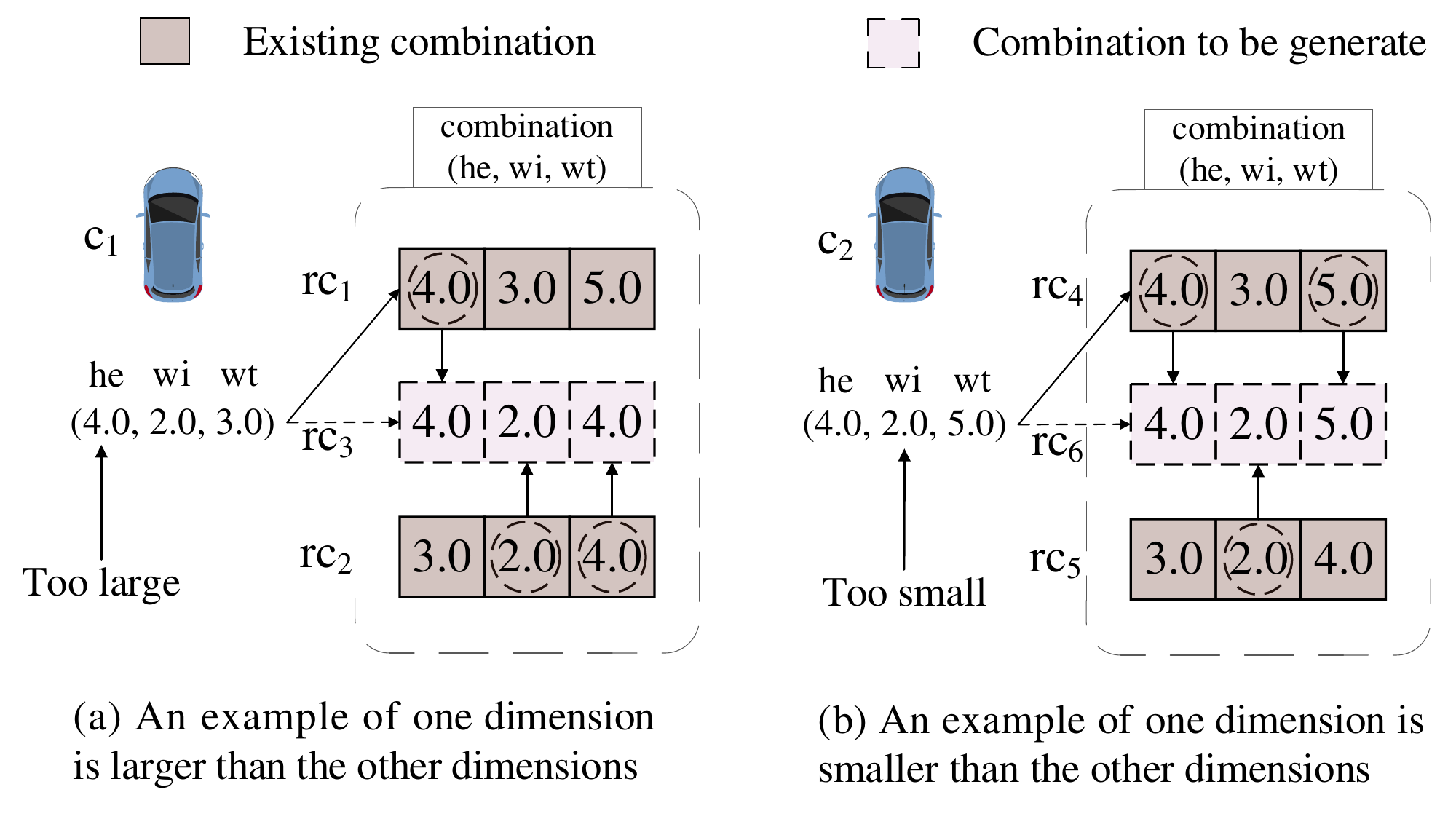}
    \caption{An example to illustrate the principle of combination rematch. The dashed boxes contain the currently existing restriction combinations, where he, wi, and wt denote height, width, and weight, respectively.}
    \label{fig:rematch}
\end{figure}

\begin{algorithm}
  \caption{Combination Rematch} 
   \label{alg:combination_rematch}
    \KwIn {$C$: cells of the road network $\{C_1,C_2,\ldots,C_n\}$\\
    $RC$: restriction combinations $\{RC_1,RC_2,\ldots,RC_n\}$} 

    
    \For{each $C_i$}{
        \For{each $rc_j$ $\in$ $RC_i$}{
         sort $RC_i$ in ascending order\\
             \For{$\mathcal{T} \in$ $\{he,wi,wt\}$}{
                $R$ $\leftarrow$ $\{rc.\mathcal{T}\mid rc \in RC_i\}$\\
                sort $R$ in ascending order\\
                $RC_i \gets RC_i$ $\cup$ $(\textsc{REMATCH}(R, RC_i, rc_j, \mathcal{T}))$
             }
        }

    }

    \SetKwFunction{rematchFunc}{REMATCH} 
    \SetKwProg{Fn}{Function}{:}{\KwRet} 

\Fn{\rematchFunc{$R, RC, rc, \mathcal{T}$}}{
    $M$ $\leftarrow$ $\varnothing$\\  
    $pos$ $\leftarrow$ index($rc.\mathcal{T}$, $R$)\\ 
     //get the position of $rc.\mathcal{T}$ in $R$\\
    \For{$j \leftarrow 1$ \KwTo $2$}{
        $rc_{temp}$ $\leftarrow$ $rc$, $rc_{temp}.i$ $\leftarrow$ $R[pos+j]$\\
        \If{$\theta(rc_{temp}, RC) \geq f* Total\_vehicles$ }{
            $M.add(rc_{temp})$
        }
        $rc_{temp}$ $\leftarrow$ $rc$, $rc_{temp}.i$ $\leftarrow$ $R[pos-j]$\\
        \If{$\theta(rc_{temp}, RC) \geq f* Total\_vehicles$ }{
            $M.add(rc_{temp})$
        }
    }
    \KwRet{$M$}
}

\end{algorithm}

For the two cases discussed above, the core idea of our optimization is to recombine the components of the different existing restriction combinations to create new combinations. Before rematch, we first sort the road restriction combinations in ascending order based on their restriction values. Then, we give two examples to illustrate our approach. For the first case, as shown in Figure 4a, for vehicle $c_1=(4.0,2.0,3.0)$ to find an acceptable path, we recombine the two combinations $rc_1=(4.0,3.0,5.0)$ and $rc_2=(3.0,2.0,4.0)$ that are adjacent to each other after sorting. Since $c_1$ is larger in height than in width and weight, we recombine the height of the larger combination with the width and weight of the smaller combination, i.e., we recombine the height of $rc_1.he$ with the $rc_2.wi$ and $rc_2.wt$ to get a more appropriate combination $rc_3=(rc_1.he,rc_2.wi,rc_2.wt)$,\ie$(4.0,2.0,4.0)$. For the second case, as shown in Fig. 4b, since the width of the vehicle $c_2=(4.0,2.0,5.0)$ is smaller compared to the height and weight, we recombine the height and weight of $rc_4=(4.0,3.0,5.0)$ and the width and weight of $rc_5=(3.0,2.0,4.0)$ to get a more appropriate combination $rc_6=(rc_4.he,rc_5.wi,rc_4.wt)$,\ie $(4.0,2.0,5.0)$. 

The details are shown in Algorithm \ref{alg:combination_rematch}. This algorithm takes the cells of road network $\{C_1,C_2,\ldots,C_n\}$ and the restriction combination sets $\{RC_1,RC_2,\ldots,RC_n\}$ as input. In lines 2-7 of the algorithm, each cell $C_i$ is processed individually. In lines 5-6, different types of restrictions are extracted in $R$ and sorted in ascending order. Subsequently, in line 7, the restriction combinations $rc_j$ are rematched, and the rematched combinations are added to the restriction combination set $RC_i$.
Lines 8-19 of the algorithm illustrate the process of the rematch function. Lines 12-18 handle the two cases discussed earlier, with lines 13-15 addressing the first case and lines 16-18 addressing the second case. In lines 14 and 17, $\theta(rc_{temp}, RC)$ is a heuristic function. Since rematching all road restriction combinations would significantly increase the number of restriction combinations, this function aims to filter out those combinations that do not contribute to improving path planning efficiency, thus controlling the number of combinations. The function estimates the number of vehicles that can plan their paths using the \textit{shortest feasible shortcut path} through the rematched road restriction combinations, using this estimate as the basis for filtering. The parameter \(f\) is set to 0.03 in this paper, and \(\text{Total\_vehicles}\) denotes the total number of vehicles in the traffic flow within the cell being processed.



                
    


\stitle{Building Shortcuts}
After completing the filtering of road restriction combinations, we build shortcuts for these retained road restriction combinations.
In each partition, we build a shortcut for each remained road restriction combination between every two boundary vertices, \ie we run modified Dijkstra's algorithm for each restriction combination between every two boundary vertices of a partition, and finally we save the shortest path we calculated by Dijkstra's algorithm as a shortcut. 

The details are shown in Algorithm \ref{alg:building_shortcuts}. This algorithm takes the cells of road network $\{C_1,C_2,\ldots,C_n\}$ and the restriction combination sets $\{RC_1,RC_2,\ldots,RC_n\}$ as input. The algorithm performs a shortcut building operation for each cell $C_i$ in lines 3-9. Specifically, lines 4-8 compute the \textit{shortest feasible shortcut path} for each pair of entry/exit vertices within cell \(C_i\) based on the existing restriction combinations and store this path in the shortcut \(s\).

\begin{algorithm}
  \caption{Building Shortcuts} \label{alg:building_shortcuts}
    \KwIn {$C$: cells of the road network $\{C_1,C_2,\ldots,C_n\}$\\
    $RC$: restriction combinations $\{RC_1,RC_2,\ldots,RC_n\}$}  

    $S$ $\leftarrow$ $\varnothing$\\
    \For{each $C_i$ $\in$ $C$}{

        \For{$u,v$ $\in$ $V^i_{entry/exit}$}{
            $s \leftarrow \varnothing$\\
             \For{each $rc_j$ $\in$ $RC_i$}{
             
                $p \leftarrow$ $shortest\_path(u,v,rc_j)$\\
                //get shortest path from $u$ to $v$ under $rc_i$\\
                $s$.add($p$)\\
             }
             $S \gets$ $S \cup s$
        }

    }
        
\end{algorithm}

\section{Optimization}
In this section, we will provide a detailed explanation of the two optimizations proposed in this paper.
\eat{
In this section, we will describe in detail that how we build shortcuts and some optimizations.

\stitle{Creating Shortcuts}


Previously we grouped the vehicles by using our proposed algorithm and then use these vehicle groups to filter the road restriction combinations of each partition, finally we got the most frequently used restriction combinations of each partition. 
After that, we build shortcuts for each partition one by one. In each partition, we build a shortcut for each reserved road restriction combination between every two boundary vertices, i.e., we run Dijkstra's algorithm with restriction for each restriction combination between every two boundary vertices of a partition, and finally we save the shortest path we calculated by Dijkstra's algorithm as a shortcut. For example, as shown in Figure 5, we have created shortcuts corresponding to a set of restriction combinations between two boundary vertices.
}

\subsection{Shortcuts Merger}

Through observation and analysis, we found that shortcut paths with different restriction combinations may have the same actual roads.
Figure \ref{fig:shortcut_storage_method}a shows an example of two different shortcut paths stored by traditional storage method, these two different shortcut paths that share the same source vertex $v_7$ and destination vertex $v_6$, but have different road restriction combinations. One shortcut path has a restriction combination of (2.0, 2.4, 10), while the other has a combination of (2.5, 2.4, 10). They both follow the same actual path: $v_7 \rightarrow v_4 \rightarrow v_1 \rightarrow v_2 \rightarrow v_6$. 
Using this storage method results in redundant storage, as shortcut paths with different road restriction combinations but the same actual paths save the same actual road multiple times.
In real-world road networks, the actual paths of shortcut paths often have long distances, leading to substantially higher storage requirements compared to other shortcut path information. Therefore, minimizing the redundant storage of identical actual paths across different shortcuts can markedly reduce memory consumption.

Based on the above analysis, we propose a novel shortcut storage method that reduces memory consumption by separating the actual paths from other information of the shortcut paths and using pointers to link other information to the actual paths.
Figure \ref{fig:shortcut_storage_method}b gives an illustration of our proposed shortcut path storage method. Our method stores the source vertex, destination vertex, restriction combinations, and actual paths of shortcut paths separately, and assigns pointers to the shortcut paths to link to their corresponding actual paths. Consequently, shortcut paths with different road restriction combinations but identical actual paths only need to store the actual path once. Given the need to establish a large number of shortcut paths in real-world road networks, our proposed storage method can significantly reduce the memory consumption associated with shortcuts compared to traditional methods.

\subsection{Shortcuts Pre-sorting}
Given a shortcut $\alpha_{s,d} = \{\pi_{s,d}^{rc_0}, \pi_{s,d}^{rc_1}, \dots, \pi_{s,d}^{rc_k}\}$ and a vehicle $c=(he,wi,wt)$. When a vehicle $c$ performs a path search in a shortcut, it must identify the path with the shortest distance among all \textit{feasible paths} within the shortcuts to ensure that it finds the \textit{shortest feasible vehicle path}, thereby achieving the optimal path.
We refer to this operation as \textit{shortcut matching}.

\begin{example}
For example, assuming that there is a vehicle $c=(2.0,2.0,10.0)$. Thus,  as shown in Figure \ref{fig:road_restriction}c, shortcut path $\pi_2$ and $\pi_3$ are the matching shortcut paths for the query. Since $\pi_2$ is the matching shortcut path with the minimum distance, $\pi_2$ is best matching shortcut path.
\end{example}


\stitle{Traditional Shortcut Path Matching Method.} In the query stage, the traditional shortcut matching method fully traverses the set of shortcut paths to find the best matching shortcut path. Obviously, the time complexity of performing shortcut matching in this way is $O(n)$.

\stitle{Our Proposed Shortcut Path Matching Method.}
To speed up shortcut matching, we propose a novel shortcut path matching method. This method pre-sorts the shortcut paths in ascending order based on their distances when building them. Therefore, during shortcut matching, we only need to traverse to the first shortcut path that is feasible for the vehicle to find the optimal matching shortcut path. Using this method for shortcut path matching, the best case occurs when the first shortcut is the optimal matching shortcut, with a time complexity of \(O(1)\). The worst case occurs when the last shortcut path is the optimal matching shortcut path, with a time complexity of \(O(n)\). On average, the optimal matching shortcut path has an equal chance of being at any position in the set, and the probability for each position is \( \frac{1}{n} \). Therefore, the average number of searches is \( \frac{1}{n} \sum_{i=1}^n i \), which is \( \frac{n+1}{2} \) times. Therefore, the time complexity in this case is \(O\left(\frac{n+1}{2}\right)\). Thus, it can be seen that our proposed method has the same time complexity as the traditional method even in the worst case, thereby significantly speeding up the shortcut path matching process.

\begin{figure}
    \centering
    \includegraphics[width=3.0in]{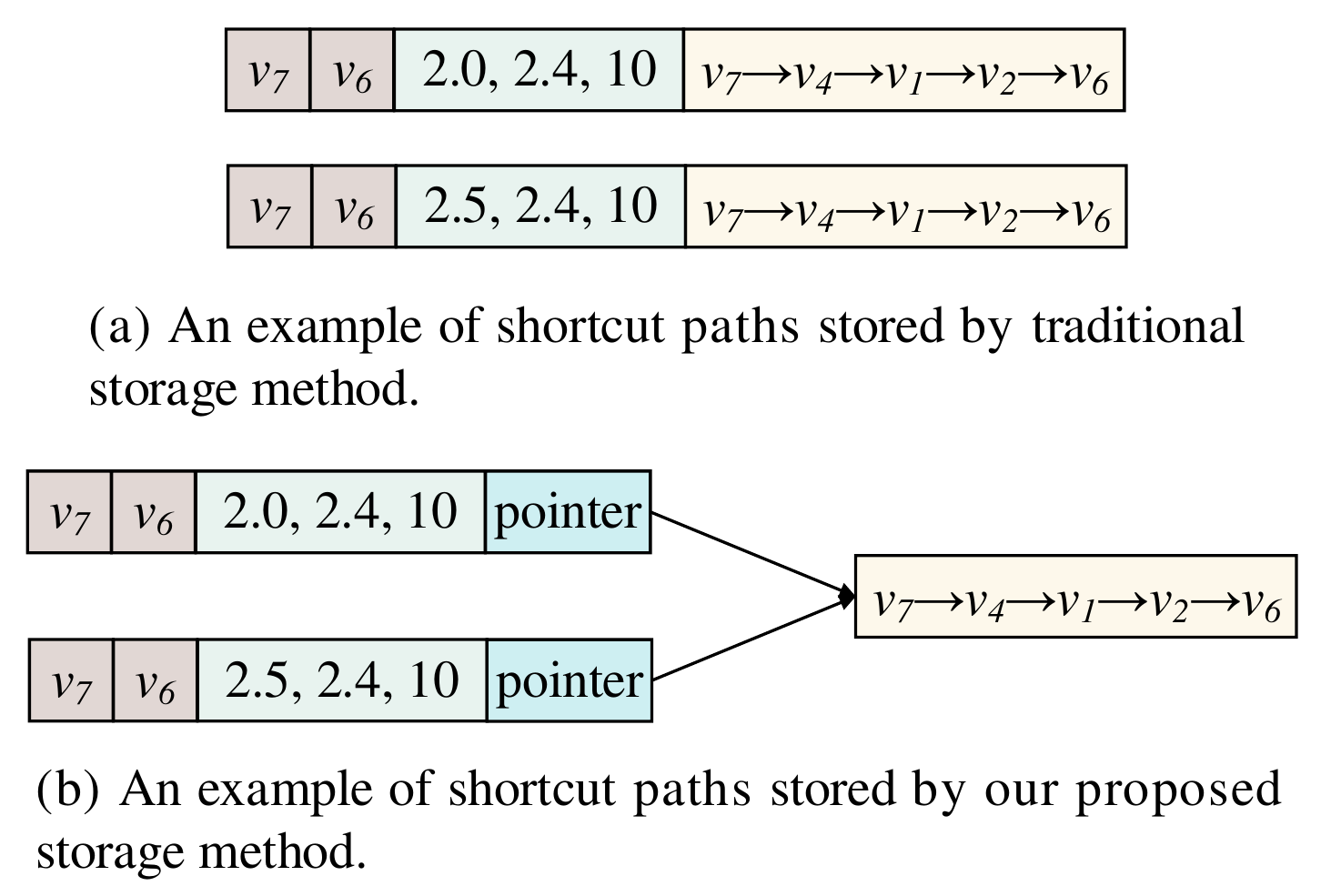}
    \caption{An illustration to show how different storage methods store shortcut paths.}
    \label{fig:shortcut_storage_method}
\end{figure}

\section{Experiments}
\label{sec:experiment}
All algorithms are implemented in 
C++ and all the experiments are conducted on a linux machine with Intel(R) Xeon(R) Gold 6248R CPU @ 3.00GHz and 96 GB RAM.


\subsection{Experimental Setup}
\stitle{Datasets.} As shown in table \ref{tab:dataset}, we use six real-world road networks which are all obtained from OpenStreetMap~\footnote{https://www.openstreetmap.org/} in our experiment: Shenyang (SY), Tokyo (TYO), Paris (PAR), London (LON), Los Angeles (LA), and San Francisco (SF).
Each road network is undirected, and consists of a set of vertices that represent intersections and a set of edges that represent a road segment. The weight on each edge denotes the distance of the road segment.
Meanwhile, we randomly assigns some road restrictions to each edge in these road networks. 
\chz{As there are currently no publicly available datasets containing road restriction information, we conducted sampling in a selected area of Shenyang, China, to collect and analyze data on local road restrictions. Based on the statistical results of this sampling, we assigned realistic road restrictions to each edge within the road networks, ensuring that the experimental data closely reflects real-world conditions.}

For traffic flow data, we generated the data based on commonly seen vehicles (high-selling vehicles) on the roads.

\stitle{Comparison Methods.} 
As previously analyzed, 
existing index-based path planning algorithms are unable to generate feasible paths in road networks with restrictions, therefore, we compare \methodname\ with the two naive methods mentioned in Section 3 and the modified Dijkstra algorithm.
\begin{itemize}[leftmargin=*, align=left]
\item \textbf{Dijkstra}: The modified Dijkstra's algorithm, which searches only the edges that are feasible for a vehicle, can obtain the \textit{shortest feasible vehicle path}.
\item \textbf{Random}: The method which randomly stores \textit{shortest feasible paths} under a certain proportion of road restriction combinations when building shortcuts.
\item \textbf{All}: The method which stores \textit{shortest feasible paths} of all road restriction combinations when building shortcuts.
\item \textbf{TRAPP}: Our proposed path planning algorithm for road networks with restrictions.

\end{itemize}

\stitle{Query Set.} For each dataset, we randomly select two vertices that do not belong to the same cell as the source and destination vertices, respectively, and randomly select a vehicle from the traffic flow data as the generation rule for a path query. Based on this rule, we generate 300 path queries as the query set.

\stitle{Evaluation Metrics.}
\chz{We use the following metrics to evaluate path planning performance: (1) \textit{Efficiency}, measured by the average time per query in path planning. (2) \textit{Space Utilization}, assessed by the number of stored paths in the shortcut structure, where lower storage usage indicates better efficiency. (3) \textit{Effectiveness}, evaluated through the Planned Path Error Rate, the Path Query Failure Rate, and the Proportion of Optimal Paths.}

\eat{
\begin{definition}[\textbf{Path error rate}]
Given a path query $Q_r = ( v_{src}, v_{des}, {c} )$, a planned path $p \in P_r$ and the optimal path $p_o$. Path error rate $\delta$ is defined as $$\delta = (dist(p)-dist(p_o))/dist(p_o)$$
\end{definition}

\begin{definition}[\textbf{Path query failure}]
Give a path query $Q_r = ( v_{src}, v_{des}, {c} )$ and a feasible path set for $Q_r$. $p$ is denoted as a planned path for $Q_r$. Path query failure is defined as 
$$dist(p)=\infty \text{ \ \ } if\ P_r \neq \varnothing $$
\end{definition}
}

\begin{table}
  \caption{Dataset description}
  \label{tab:dataset}
  \centering
  \begin{tabular}{c@{\extracolsep{30pt}}c@{\extracolsep{30pt}}c}
    \toprule
    Road Network  &  \#Vertices &  \#Edges\\
    \midrule
    Shenyang (SY) & 47,773 & 105,378 \\
    Tokyo (TYO) & 94,016 & 212,754 \\
    Paris (PAR) & 137,411 &  292,038 \\
    London (LON) & 155,221 & 326,882 \\
    Los Angeles (LA) & 161,384 & 346,764 \\
    San Francisco (SF) & 300,617 &  633,958\\
  \bottomrule
\end{tabular}
\end{table}

\begin{figure}[htbp]
  \centering
  \begin{minipage}[b]{0.232\textwidth}
    \centering
    \includegraphics[width=\textwidth]{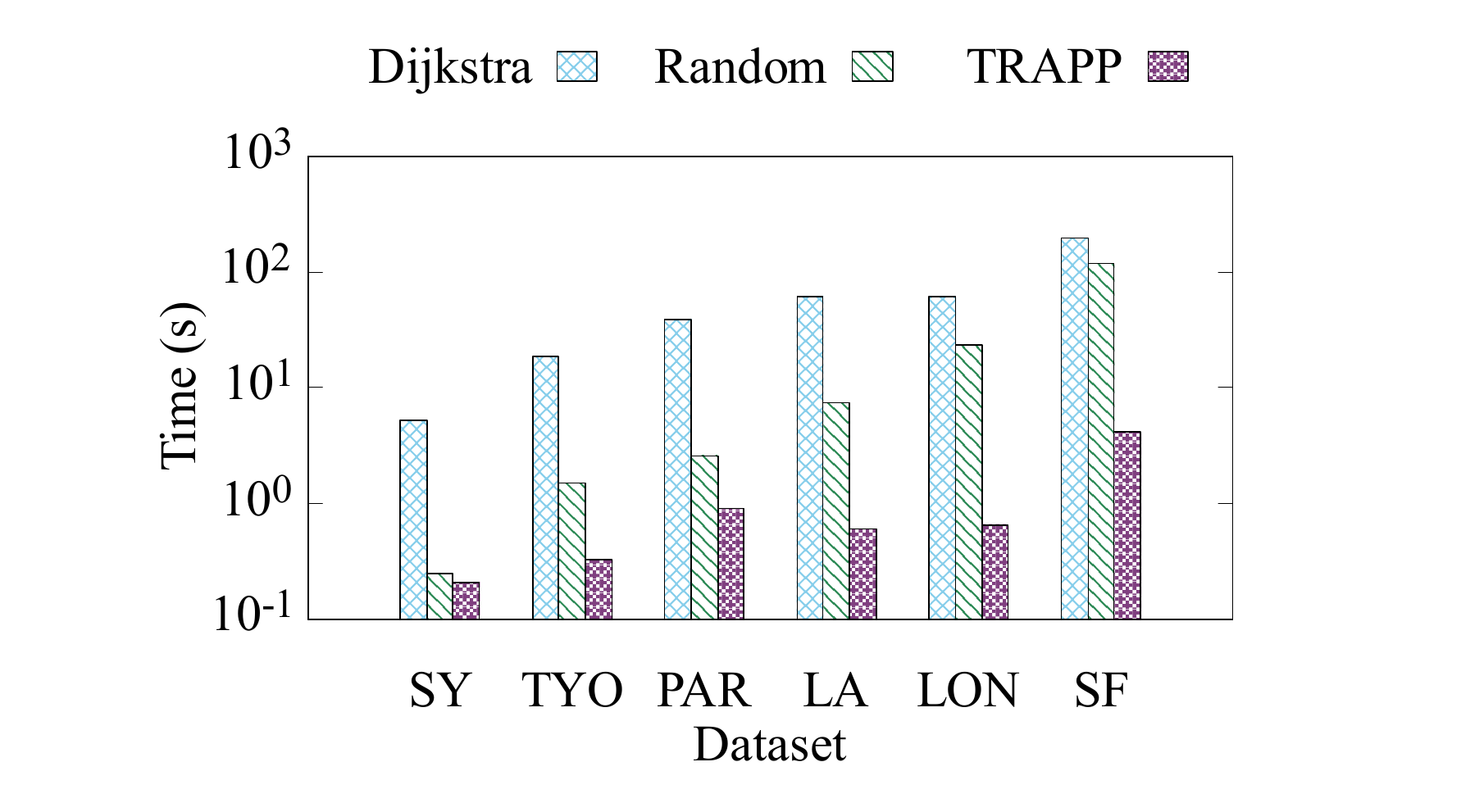} 
    \vspace{-13pt} 
    \caption{Path planning time comparison.}
    \label{fig:time_comparison}
  \end{minipage}
  \hspace{0.003\textwidth} 
  \begin{minipage}[b]{0.232\textwidth}
    \centering
    \includegraphics[width=\textwidth]{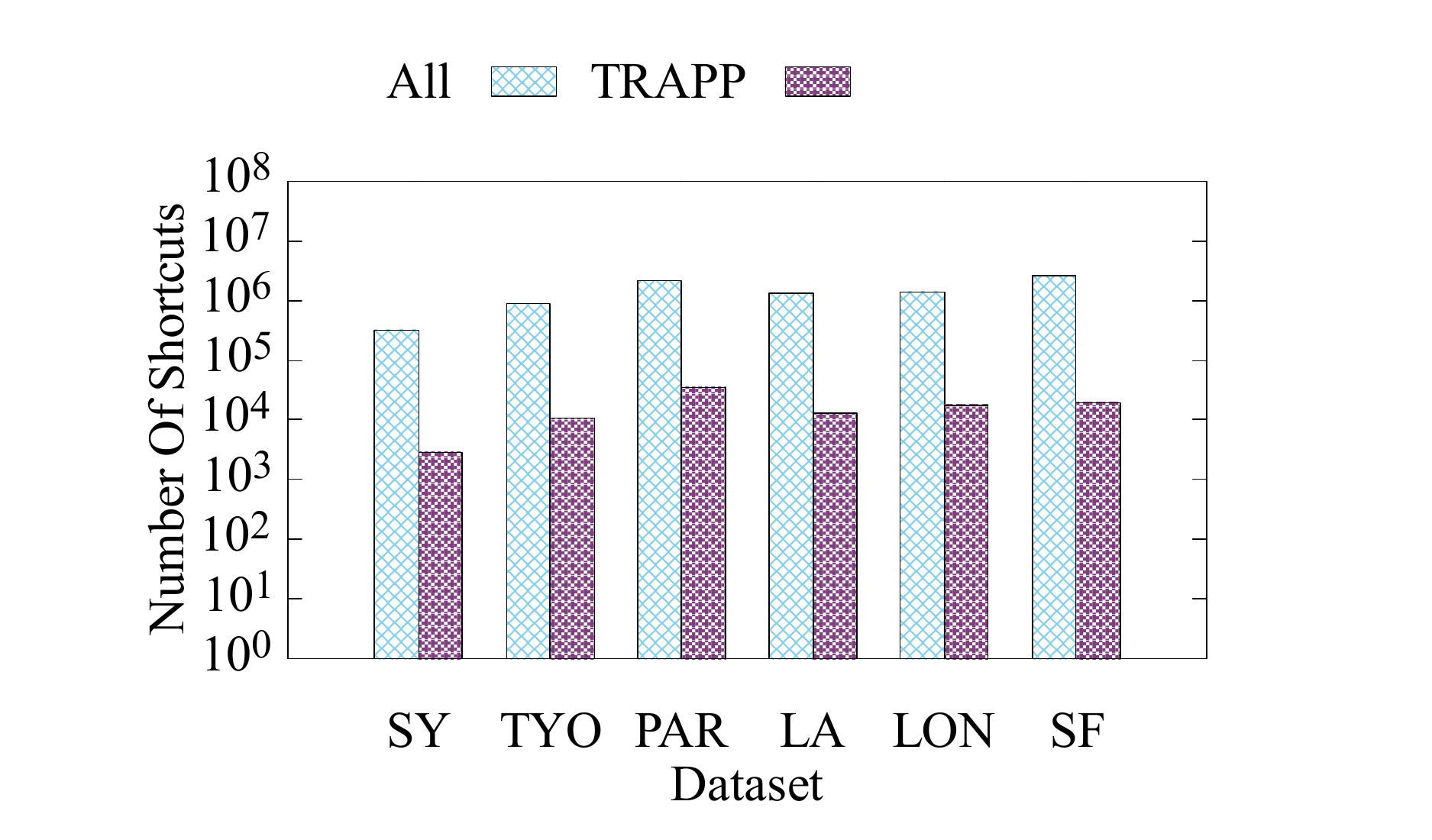} 
    \vspace{-13pt} 
    \caption{Comparison of number of shortcut paths.}
    \label{fig:storage_test}
  \end{minipage}
\end{figure}

\subsection{Path Planning Time Evaluation}
In this experiment, we evaluate the efficiency of path planning by comparing the runtime of the modified Dijkstra's algorithm, \randomMethod and our method \methodname on the six datasets: SY, TYO, PAR, LON, LA, and SF, as shown in Table \ref{tab:dataset}. For each dataset, we use the same query set for different methods to conduct comparative experiments, and take the average query time of a single query in the query set as the experimental result.

The experimental results are shown in Figure 5. From the results, we can see that for all datasets, the path planning time for \methodname is significantly lower than that of the modified Dijkstra algorithm. In the best cases, the query speed of \methodname is 426 times faster than that of Dijkstra, because Dijkstra's algorithm needs to search a large number of edges. Meanwhile, the query time for \randomMethod falls between Dijkstra and \methodname. This is because the failure rate of the \randomMethod method is much higher than that of \methodname, and when path planning is failed, \randomMethod uses the Dijkstra algorithm for searching, which leads to an increase in its query time.


In summary, our proposed method \methodname effectively ensures the efficiency of path planning.





\begin{figure}[htbp]
  \centering
  \begin{minipage}[b]{0.232\textwidth}
    \centering
    \includegraphics[width=\textwidth]{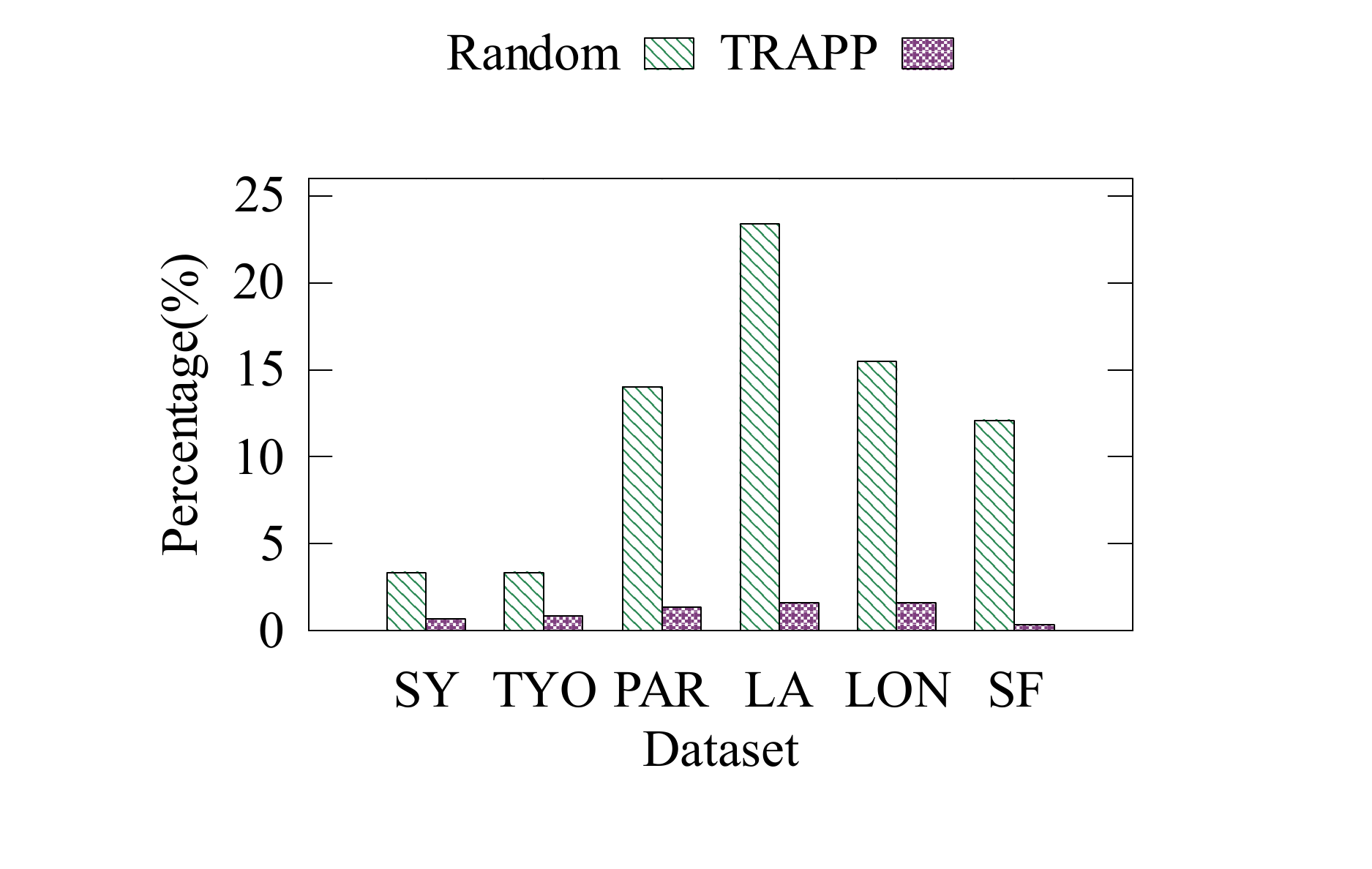} 
    \vspace{-25pt} 
    \caption{Comparison of path error rate.}
    \vspace{-1pt}
    \label{fig:path_error_test}
  \end{minipage}
  \hspace{0.003\textwidth} 
  \begin{minipage}[b]{0.232\textwidth}
    \centering
    \includegraphics[width=\textwidth]{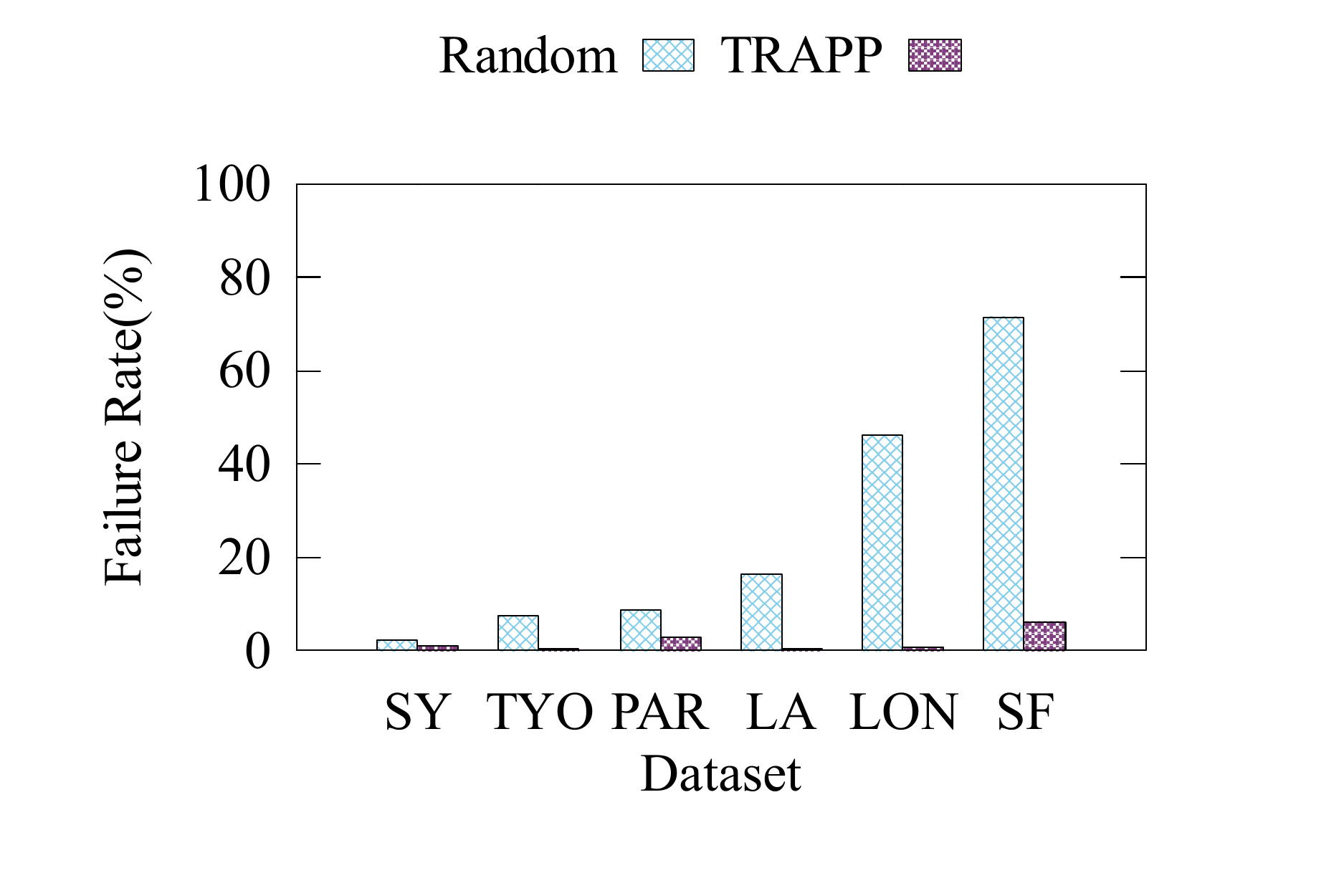} 
    \vspace{-25pt} 
    \caption{Comparison of failure rate.}
    \vspace{-1pt}
    \label{fig:failure_rate}
  \end{minipage}
  
\end{figure}

\subsection{Shortcut Storage Evaluation}
To demonstrate that the shortcuts built by \methodname can effectively control the number of actual shortest paths, we conduct experiments to test the number of shortcuts built by our proposed method \methodname and compare it with \allShortcutMethod on six different datasets.

As shown in Figure \ref{fig:storage_test}, in all datasets, the number of actual shortest paths of shortcuts built by \methodname is significantly less than that of \allShortcutMethod. In the best case, the number of actual shortest paths preserved by \methodname is only 0.6\% of those preserved by \allShortcutMethod.\chz{} These experimental results show that \methodname effectively reduces the number of actual shortest paths of built shortcuts, thereby decreasing the computational and memory overhead.



\begin{figure}[htbp]
  \centering
  \begin{minipage}[b]{0.232\textwidth}
    \centering
    \includegraphics[width=\textwidth]{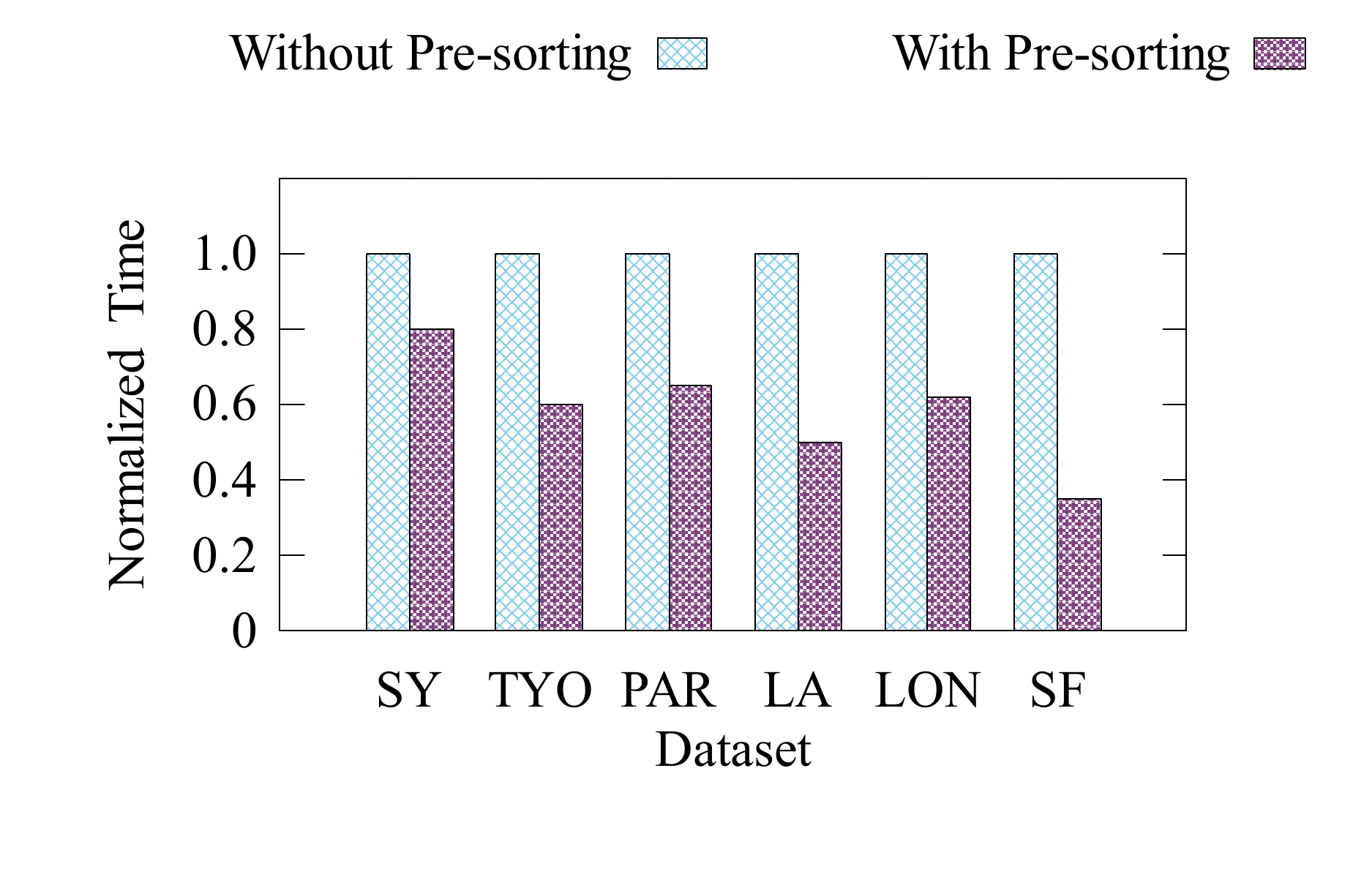} 
    \vspace{-17pt} 
    \caption{Shortcut path matching time comparison.}
    \label{fig:matching_time}
  \end{minipage}
  \hspace{0.003\textwidth} 
  \begin{minipage}[b]{0.232\textwidth}
    \centering
    \includegraphics[width=\textwidth]{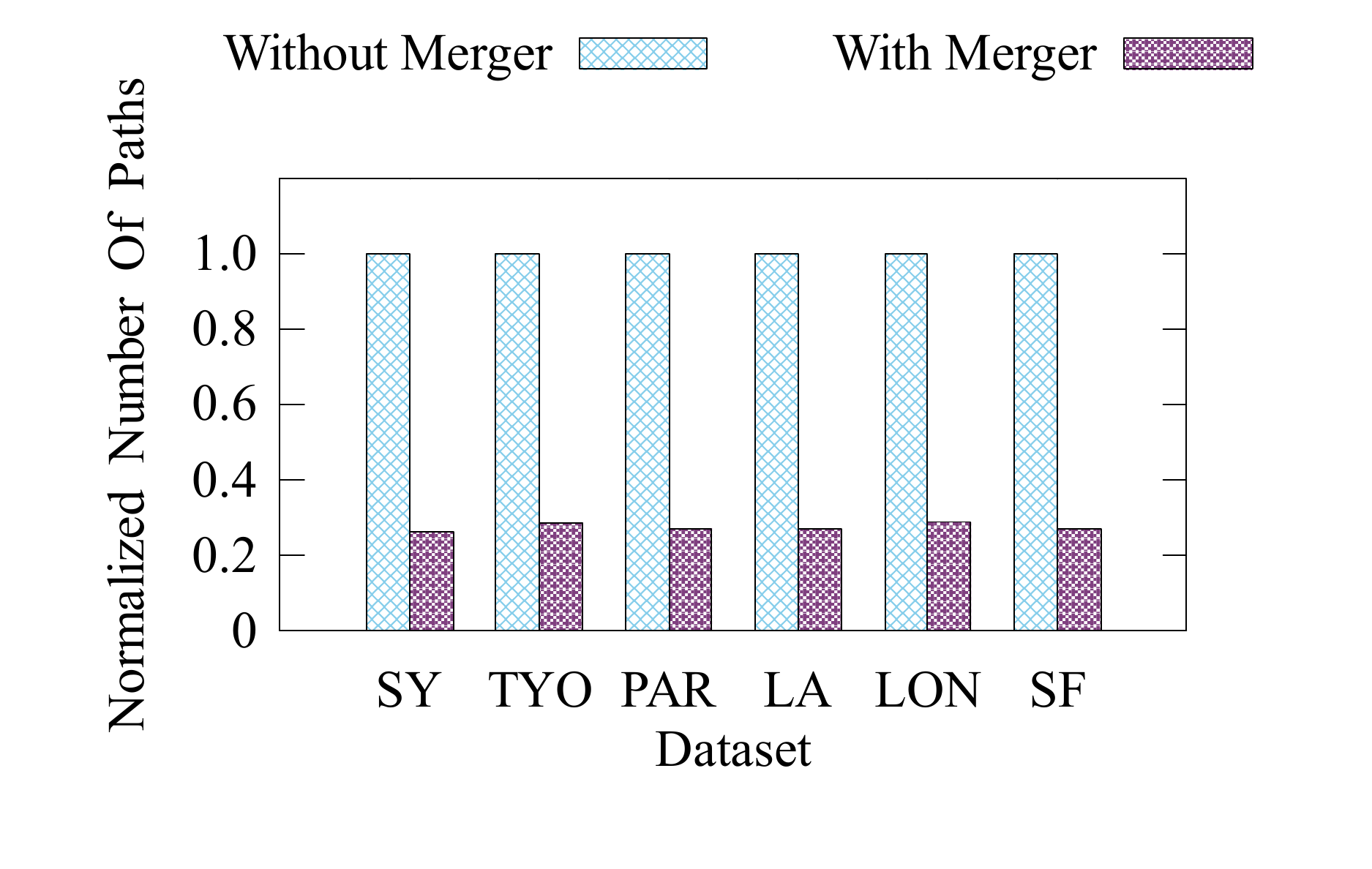} 
    \vspace{-17pt} 
    \caption{Comparison of the number of actual paths.}
    \label{fig:actual_path_number}
  \end{minipage}
\end{figure}
\subsection{Planned Path Error Rate Evaluation}


To demonstrate the effectiveness of \methodname in path planning, we conduct experiments comparing the planned path error rate of \methodname and \randomMethod. The planned path error rate is defined as the error rate between the path planned by the algorithm and the optimal path for a given path planning query $q$, \ie
$$\delta = (dist(p)-dist(p^*))/dist(p^*)$$
where $p$ is the planned path and $p^*$ is the optimal path for query $q$.
In these experiments, Dijkstra's algorithm serves as the benchmark to obtain the feasible shortest path. For each dataset, we use the same query set across different methods and calculate the average planned path error rate of each individual query within the set, using this as the experimental result. 

As shown in Figure \ref{fig:path_error_test}, the planned path error rate of \methodname is significantly lower than that of \randomMethod. In nearly all datasets, the average path distances planned by \methodname are within 5\% of the shortest path.
This demonstrates the excellent path planning performance of \methodname, allowing the majority of vehicles to get an acceptable path even when the optimal path is not found.

\subsection{Comparison of Path Query Failure Rate}
Path query failure occurs when a path query theoretically has at least one solution but cannot find one during practical execution. The path query failure rate is defined as the ratio of the number of failed queries to the total number of queries for a given query set $Q$.
A higher path query failure rate indicates a less effective method. To validate the efficacy of our proposed approach, we conduct a comparative analysis of path query failure rates between \methodname and \randomMethod in our experiment.

\chz{}As shown in Figure \ref{fig:failure_rate}, the failure rate of \methodname is significantly lower than that of \randomMethod. In all the datasets, the failure rate of \methodname is within 5\%, and in the best cases, the failure rate of \methodname can be below 1\%. Our experimental results indicate that \methodname can plan feasible paths for the vast majority of vehicles.



\subsection{Shortcut Pre-sorting Evaluation}



To verify the effectiveness of shortcut pre-sorting, we conduct an experiment to evaluate the matching time of shortcuts with setting $K=30$. we compare the shortcuts matching time of the traditional matching method with that of our proposed matching method. 

Figure \ref{fig:matching_time} shows the normalized result of shortcuts matching time, our method reduces the time by 20\% to 75\% compared to the traditional method. This significant reduction demonstrates that the shortcut pre-sorting method can substantially improve the processing efficiency of path queries, thereby enhancing overall performance in practical applications.

\subsection{Shortcut Merger Evaluation}



We evaluated the efficiency of our proposed shortcut storage method by comparing its number of stored paths with that of the traditional shortcut storage method. 

After normalization, as shown in Figure \ref{fig:actual_path_number}, across all datasets, our proposed method significantly reduces redundant actual path storage by 60\% to 80\% compared to the traditional method. This substantial reduction in memory consumption shows the effectiveness and efficiency of our proposed storage method, making it suitable for application in large-scale road networks.

\begin{figure}
    \centering
    \includegraphics[width=0.7\linewidth]{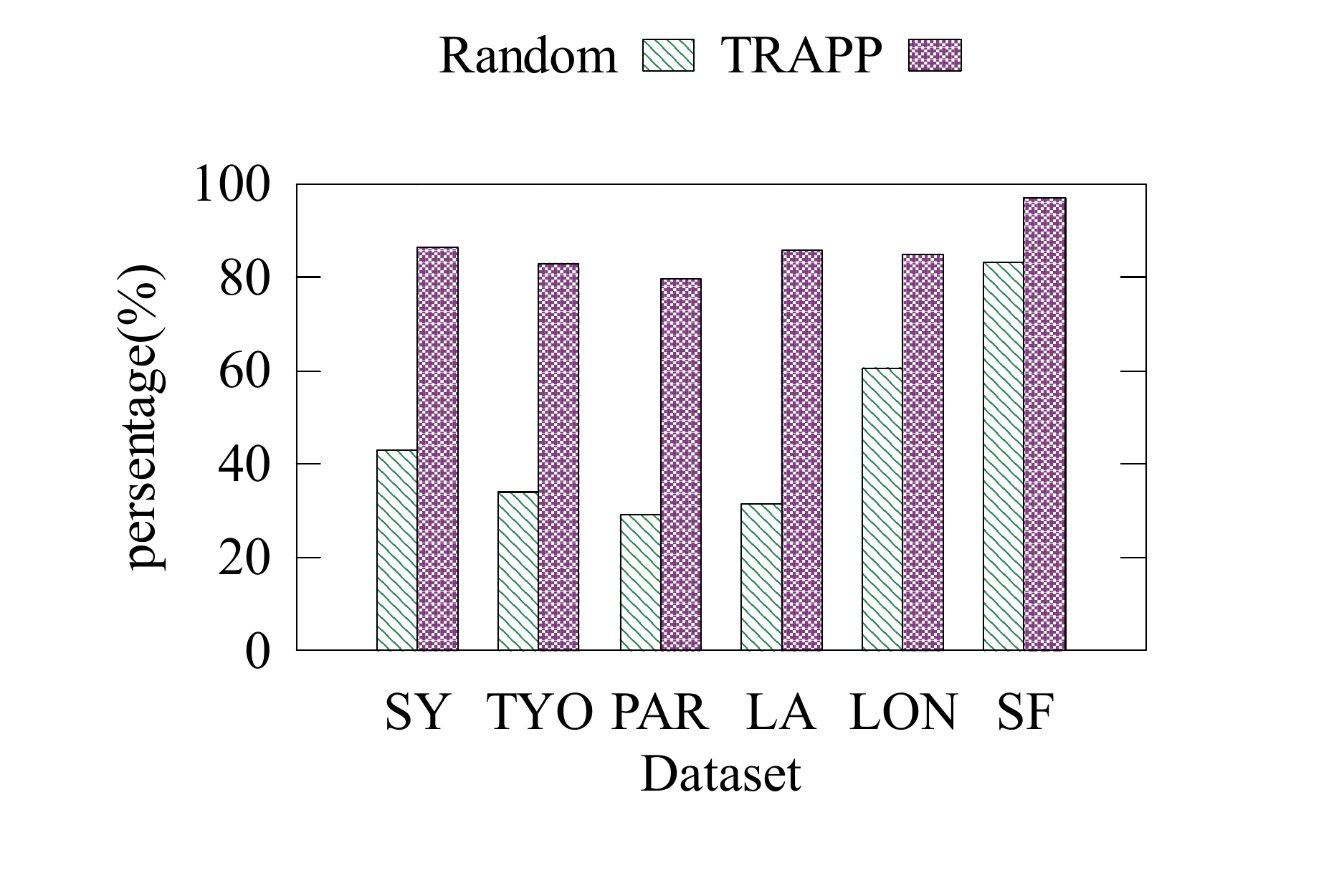}
    \vspace{-15pt}
    \caption{Proportion of Optimal Paths.}
    \vspace{-15pt}
    \label{fig:Optimal_Paths}
\end{figure}

\begin{figure*}[t] 
  \centering
  \begin{minipage}[t]{0.3\textwidth}
    \centering
    \includegraphics[width=\textwidth]{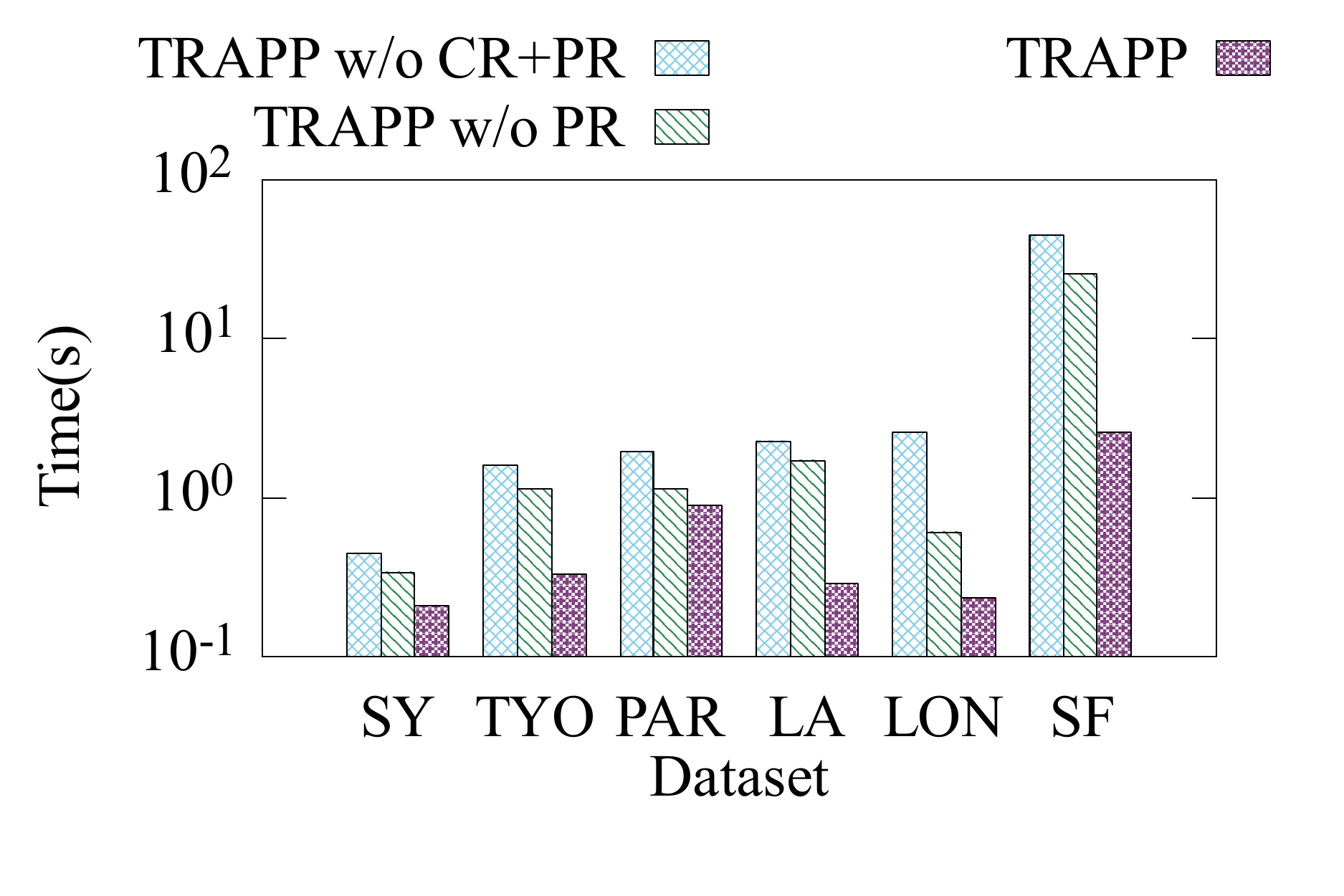} 
    \vspace{-22pt} 
    \caption*{(a) Path planning time.}
    \label{fig:ablation_time}
  \end{minipage}
  \begin{minipage}[t]{0.3\textwidth}
    \centering
    \includegraphics[width=\textwidth]{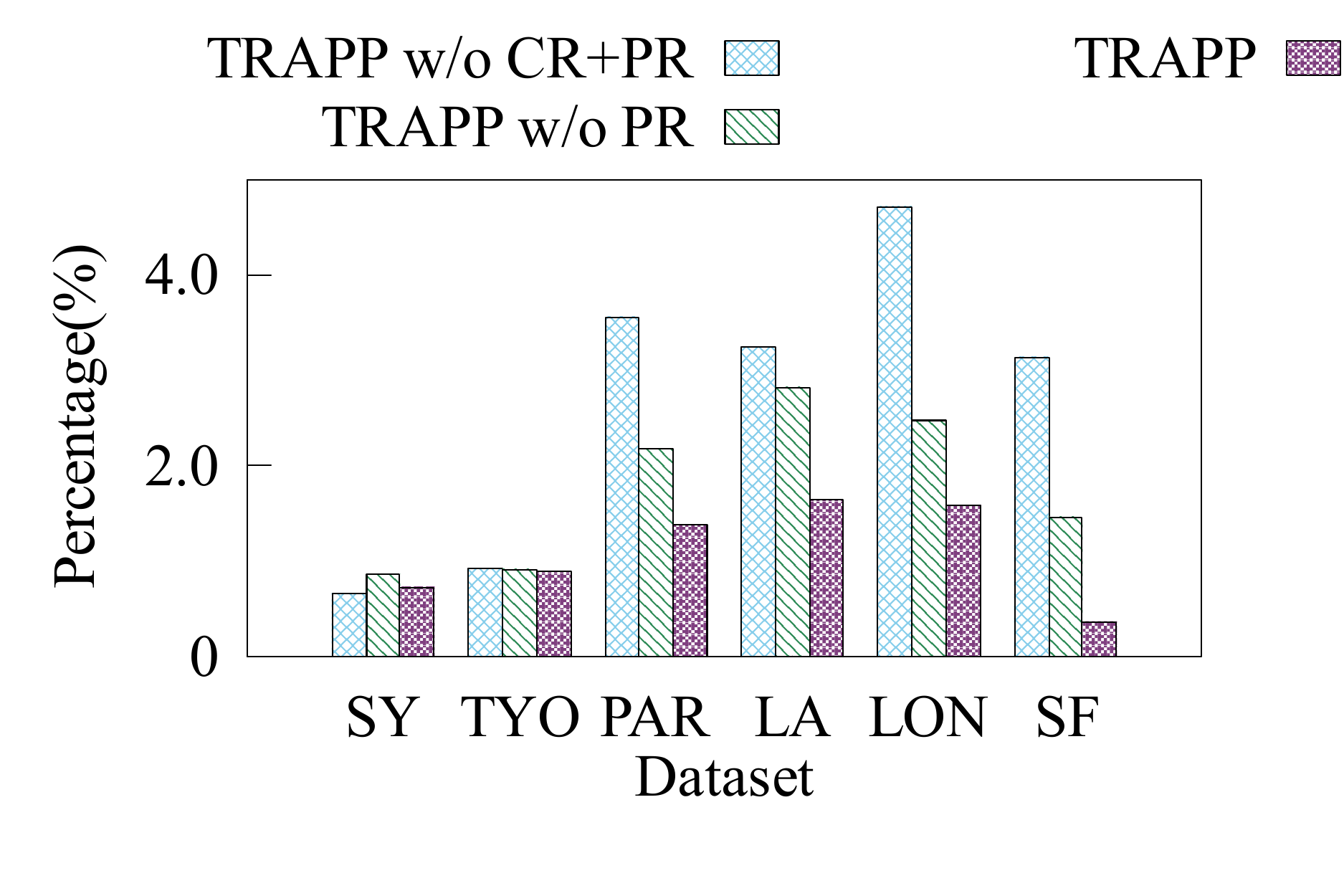} 
    \vspace{-22pt} 
    \caption*{(b) Path error rate.}
    \label{fig:ablation_error}
  \end{minipage}
  \begin{minipage}[t]{0.3\textwidth}
    \centering
    \includegraphics[width=\textwidth]{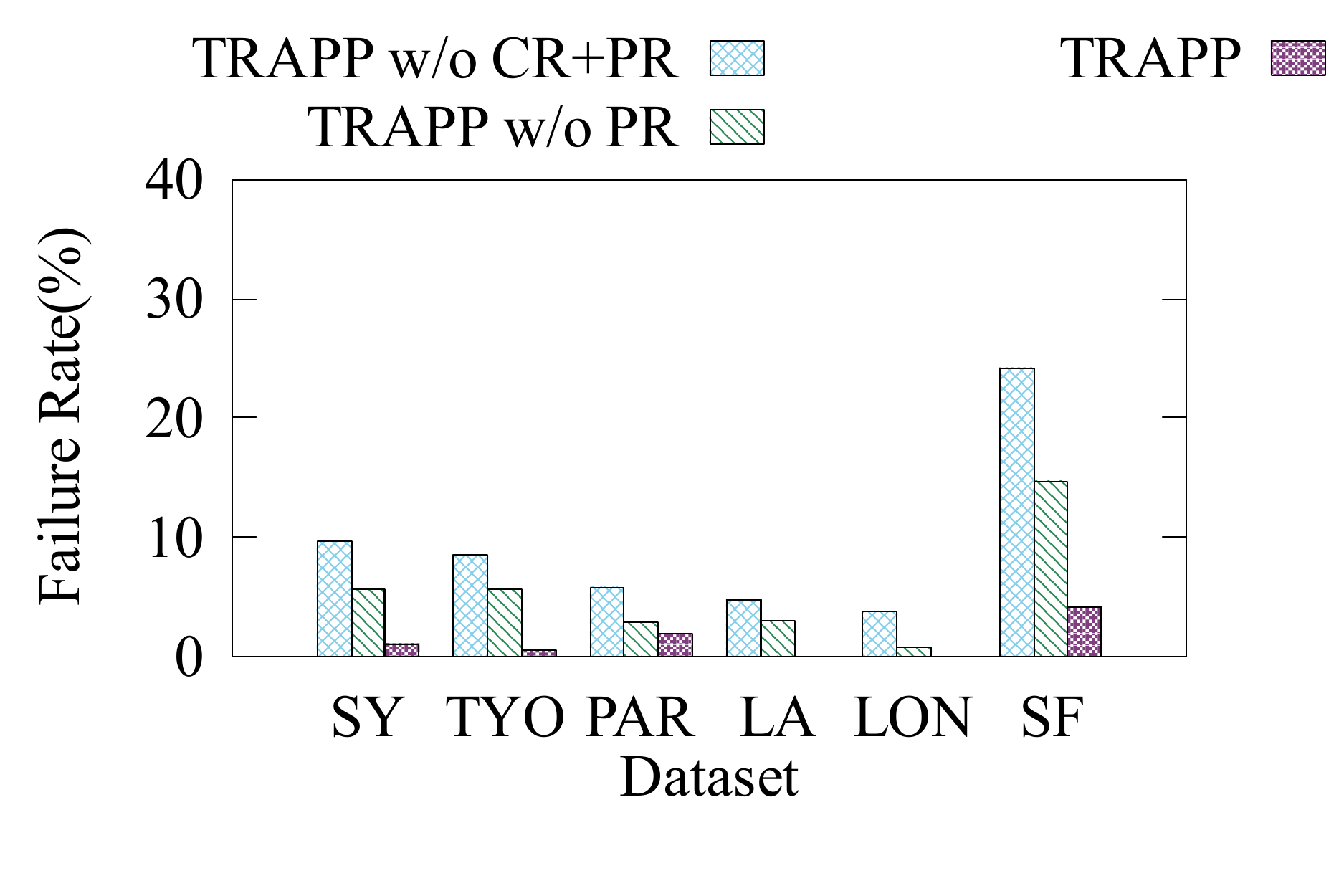} 
    \vspace{-22pt} 
    \caption*{(c) Failure rate.}
    \label{fig:ablation_failure}
  \end{minipage}
  \begin{minipage}[t]{0.65\textwidth} 
    \centering
    \begin{minipage}[t]{0.45\textwidth} 
      \centering
      \includegraphics[width=\textwidth]{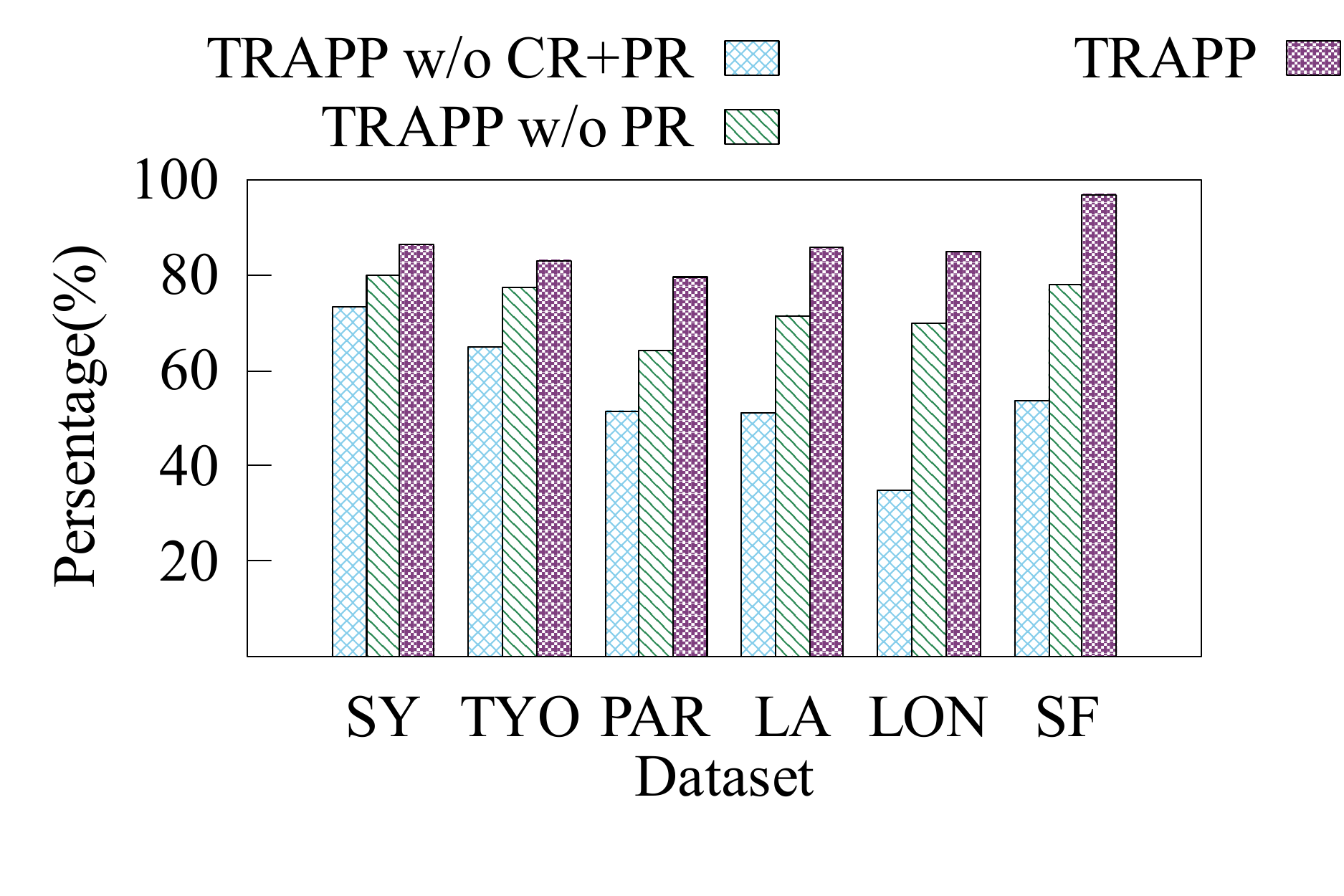} 
      \vspace{-22pt} 
      \caption*{(d) Proportion of Optimal Paths.}
      \label{fig:ablation_optimal}
    \end{minipage}
    \hfill
    \begin{minipage}[t]{0.45\textwidth} 
      \centering
      \includegraphics[width=\textwidth]{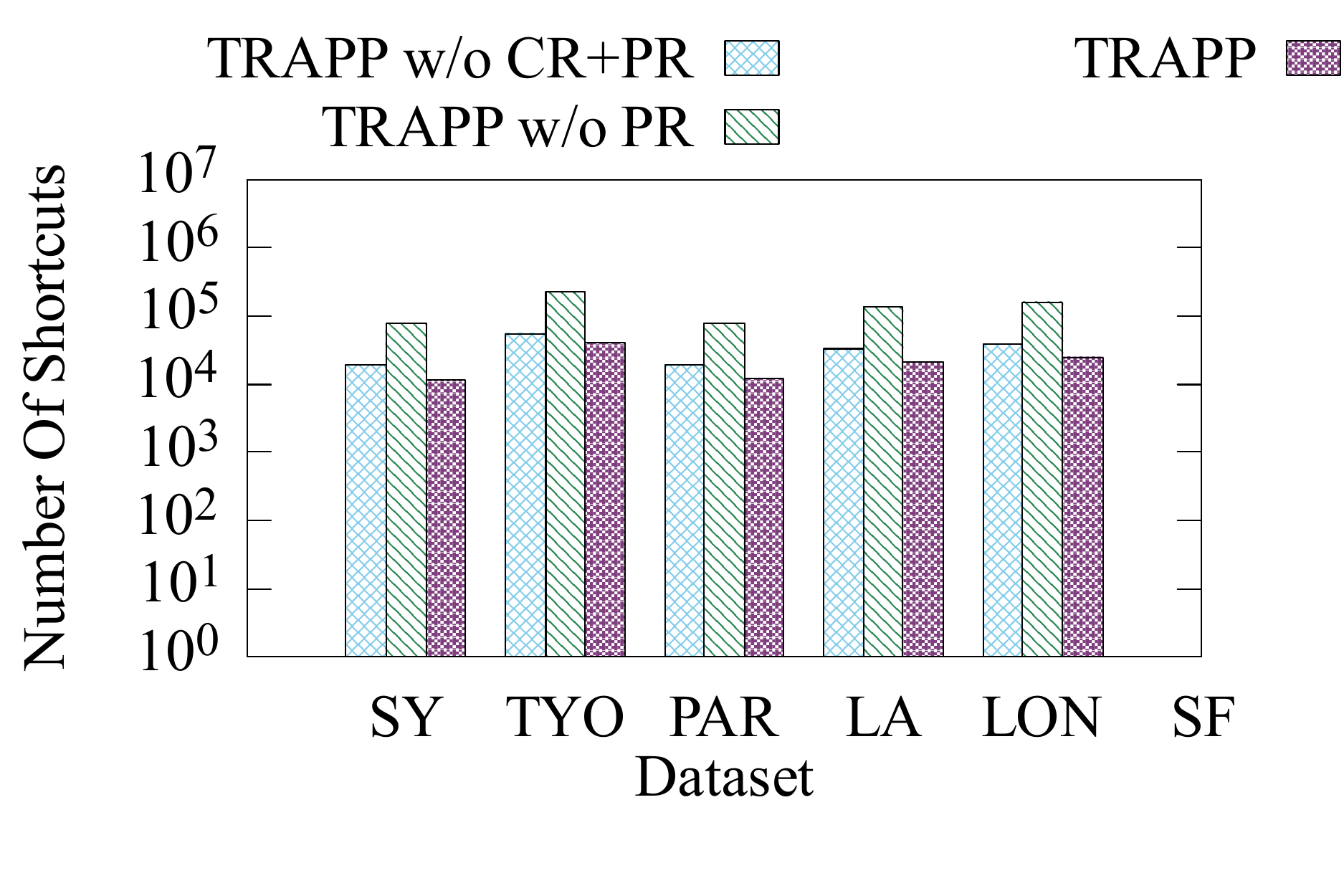} 
      \vspace{-22pt} 
      \caption*{(e) Comparison of number of shortcut paths.}
      \label{fig:ablation_shortcut_usage}
    \end{minipage}
  \end{minipage}
  \caption{The efficiency of different optimizations.
  }
  \label{fig:ablation_exp}
\end{figure*}

\subsection{Evaluation of the Proportion of Optimal Paths}
To comprehensively evaluate our proposed method \methodname, we conduct experiments on six different road networks to measure the optimal solution proportion of \methodname. The optimal solution proportion is defined as the proportion of the number of queries that achieve the optimal solution to the total number of queries for a given query set $Q$. A higher proportion of optimal paths indicates more frequent planning of optimal solutions, thereby demonstrating superior path planning performance. 

As shown in Figure \ref{fig:Optimal_Paths}, across different datasets, the optimal solution proportion of our proposed \methodname method remains consistently high. In the best-case scenario, the optimal solution proportion reaches 96.8\%, while even in the worst-case scenario, it \eat{maintains}remains at 79.6\%. This demonstrates the excellent path planning effectiveness of \methodname. Additionally, in all datasets, the optimal solution proportion of \methodname is significantly higher than that of \randomMethod. For instance, in the SF dataset, the optimal solution proportion of \methodname is nearly seven times higher than that of \randomMethod, despite both methods preserving a similar number of shortest paths. This clearly illustrates that our proposed method effectively filters out rarely used road restriction combinations and retains more frequently used ones, thereby ensuring excellent path planning performance with minimal computational and memory overhead.

\subsection{Ablation Study}

To verify the contribution of each component in \methodname, we conduct experiments on the following variants of \methodname.
\begin{itemize}[leftmargin=*, align=left]
\item \textbf{\methodname w/o PR}: This variant of \methodname excludes the Partition-aware Refinement (PR) component.
\item \textbf{\methodname w/o CR+PR}: This variant of \methodname excludes the Partition-aware Refinement (PR) and Combination Rematch (CR) components, essentially utilizing only the k-means algorithm.
\end{itemize}


The ablation experiments evaluated the path planning time, path distance, failure rate, optimal solution ratio, and the number of preserved actual shortest paths for \methodname, \methodname w/o PR, and \methodname w/o CR+PR across different datasets. Figure \ref{fig:ablation_exp}a shows that \methodname consistently achieved the shortest path planning times across all datasets. This is because \methodname has the lowest path planning failure rate among the three methods, resulting in fewer regressions to the Dijkstra algorithm. Figure \ref{fig:ablation_exp}b indicates that all three methods maintain low path planning errors when the optimal solution cannot be found. Figure \ref{fig:ablation_exp}c demonstrates that the complete \methodname has the lowest failure rate, significantly lower than \methodname w/o CR and \methodname w/o CR+PR. Additionally, \methodname w/o CR also has a significantly lower failure rate than \methodname w/o CR+PR, highlighting the effectiveness of the CR and PR components. Figure \ref{fig:ablation_exp}d presents the experimental results of the optimal solution ratio for the three methods, showing that \methodname has a significantly higher optimal solution ratio compared to \methodname w/o CR and \methodname w/o CR+PR, with the performance improving as more components are included. Figure \ref{fig:ablation_exp}e shows that \methodname preserves the fewest actual shortest paths, while \methodname w/o CR preserves the most, because the PR component retains more road restriction combinations. The experimental results also demonstrate that the PR component effectively controls the number of preserved road restriction combinations.

In summary, the experimental results of the ablation experiments demonstrate that the CR and PR components significantly enhance both the effectiveness and efficiency of \methodname.
\section{Related Work}
\label{sec:related_work}
\stitle{Non-index-based Path Planning Algorithms}. Dijkstra's algorithm \cite{Dijkstra59} is a classic shortest path algorithm that uses a greedy strategy to perform a full graph search to obtain the single-source shortest path on a non-negative weighted graph. Bellman-Ford algorithm~\cite{bellman58,ford56} is capable of solving the shortest path problem in graphs with negative edge weights. Johnson algorithm~\cite{johnson1977efficient} is a combination of the previous two algorithms. Floyd-Warshall~\cite{floyd1962algorithm,warshall1962theorem} algorithm utilizes dynamic programming for path planning and A* search algorithm~\cite{AstarSearch} incorporates heuristic methods to reduce unnecessary searches.~\cite{non_index_experiment_1} provided a theoretical summary of representative non-index-based algorithms and conducted an experimental evaluation.

\stitle{Index-based Path Planning Algorithms}.
Graph indexing has been widely studied to enhance efficiency. Index-based methods typically involves pre-computing the shortest paths between some vertices and building indices to accelerate path planning. CH~\cite{CH} and CRP~\cite{CRP} are the state-of-the-art path planning methods for road networks, where CH builds shortcuts for the neighbors of the contracted vertices in a specific order and CRP constructs a multi-level graph structure based on the given road network and builds shortcuts between the entry and exit vertices of each cell at each level. AH~\cite{AH}, which is similar to CH, leverages spatial information to construct a hierarchical structure and build multiple shortcuts within it. EDP~\cite{EDP} builds dynamic indices within the graph, enabling efficient edge-constrained shortest path queries on dynamic graphs. G*-tree~\cite{G*-tree} is an efficient hierarchical indexing structure that offers advantages such as good scalability. Xiao, Yanghua, et al.~\cite{symmetry} design a novel path planning algorithm, which significantly reduces the indexing space consumption by exploiting the widespread symmetry in the graph, while ensuring the efficiency of path planning. H2H-index~\cite{H2H}, which combines the advantages of 2-hop labeling and hierarchy, builds indices based on tree decomposition and maintains a hierarchical structure among all vertices in the graph. This method addresses the problem of large search spaces in hierarchical indexing structures for long-distance queries, while also mitigating the high computational overhead in hop-based indexing schemes for short-distance queries. Zhang Y, et al.~\cite{CHandH2Hsubboundedness} introduces the concept of relative subboundedness, studies the boundedness of CH and H2H, and proposes an incremental algorithm IncH2H, which can quickly update indices on dynamic road networks. Inspired by H2H, the LSD-index~\cite{LSD-index}, which is based on tree decomposition, can perform label-constrained shortest path queries in a very short time while also reducing the memory consumption of the indices. Building upon the shortest-distance query processing methods PLL~\cite{PLL} and CTL~\cite{CTL}, ~\cite{MLL} extends their capabilities to efficiently handle shortest-path queries and introduces a novel path planning method, MLL, which achieves a significantly faster query speed.
Some literature~\cite{path_planning_experiment_1,path_planning_experiment_2} conduct evaluations on the performance of the representative path planning algorithms.

Real-world road networks have numerous different restrictions, yet existing algorithms do not take these into account. Therefore, it is crucial to study path planning algorithms that consider road network restrictions. \methodname addresses this issue effectively, enabling the planning of passable, optimal, and acceptable paths for the majority of vehicles in road networks with restrictions, while significantly reducing the computational and storage overhead of the indices.

\section{CONCLUSION}
\label{sec:conclusion}
In this paper, we define the path planning problem under road restrictions and analyze the combinatorial explosion problem in existing methods applied to road networks with restrictions. Subsequently, we propose \methodname, a path planning algorithm tailored for these networks. The core idea of \methodname is to utilize traffic flow information within the road network to filter road restriction combinations, thereby addressing the substantial computational and memory overhead caused by the combinatorial explosion. Experimental results validate the effectiveness of \methodname in path planning, as well as its low memory and computational overhead.

\bibliographystyle{ACM-Reference-Format}  
\bibliography{sample} 
\end{document}